\documentclass[11pt]{article}
\usepackage[utf8]{inputenc}
\usepackage[margin=1in]{geometry}
\usepackage{amsmath}
\usepackage{amssymb}
\usepackage{amsfonts}
\usepackage{amsthm}
\usepackage{comment}
\usepackage{microtype}
\usepackage{mathrsfs}
\usepackage{hyperref}
\usepackage{ytableau}
\usepackage{caption}

\theoremstyle{definition}
\newtheorem{thm}{Theorem}[section]
\newtheorem{lem}[thm]{Lemma}
\newtheorem{df}[thm]{Definition}
\newtheorem{dfprop}[thm]{Proposition/Definition}
\newtheorem{prop}[thm]{Proposition}
\newtheorem{cor}[thm]{Corollary}

\newtheorem{rem}[thm]{Remark}

\newcommand{\Q}{\mathbb{Q}}

\newcommand{\Z}{\mathbb{Z}}
\newcommand{\N}{\mathbb{N}}
\newcommand{\Y}{\mathbb{Y}}
\newcommand{\PP}{P^{\textrm{HL}}}
\newcommand{\QQ}{Q^{\textrm{HL}}}
\newcommand{\F}{F^{\textrm{HL}}}
\newcommand{\Fcal}{\mathcal{F}^{\textrm{HL}}}
\newcommand{\EE}{E^{\textrm{HL}}}
\newcommand{\HH}{H^{\textrm{HL}}}
\newcommand{\FF}{\mathbb{F}}
\newcommand{\inv}{\textrm{inv}}
\newcommand{\Tab}{\textrm{Tab}}

\definecolor{ashgrey}{rgb}{0.7, 0.75, 0.71}
\definecolor{battleshipgrey}{rgb}{0.52, 0.52, 0.51}

\numberwithin{equation}{section}

\begin{document}

\title{Interpolation Macdonald operators at infinity}

\author{Cesar Cuenca}

\date{}

\maketitle

\abstract{
We study the interpolation Macdonald functions, remarkable inhomogeneous generalizations of Macdonald functions, and a sequence $A^1, A^2, \ldots$ of commuting operators that are diagonalized by them.
Such a sequence of operators arises in the projective limit of finite families of commuting q-difference operators studied by Okounkov, Knop and Sahi.
The main theorem is an explicit formula for the operators $A^k$.
Our formula involves the family of Hall-Littlewood functions and a new family of inhomogeneous Hall-Littlewood functions, for which we give an explicit construction and identify as a degeneration of the interpolation Macdonald functions in the regime $q \rightarrow 0$.
This article is inspired by the recent papers of Nazarov-Sklyanin on Macdonald and Sekiguchi-Debiard operators, and our main theorem is an extension of their results.
}

\tableofcontents

\section{Introduction}

The theory of Macdonald functions and their associated operators has recently enjoyed considerable attention and found applications in various areas, see for example, \cite{AK, FHHSS, GZ, N, NS3, NS4} and references therein.
In the present work, we study an \textit{inhomogeneous} version of the Macdonald functions, called the \textit{interpolation Macdonald functions}.
More explicitly, for any partition $\lambda$, one can construct a symmetric function $I_{\lambda}(\cdot; q, t)$ with coefficients in the field $\Q(q, t)$ and such that its top-degree homogeneous term is the Macdonald function parametrized by $\lambda$.
The interpolation Macdonald function $I_{\lambda}(\cdot; q, t)$ is certain limit of the interpolation (or \textit{shifted}) Macdonald polynomials $I_{\lambda|N}(x_1, \ldots, x_N; q, t)$ considered previously by Knop \cite{K}, Okounkov \cite{Ok3, Ok2, Ok}, Olshanski \cite{Ol2, Ol}, Sahi \cite{S}, among others.

In \cite{Ok}, a finite hierarchy of $q$-difference operators diagonalizing the interpolation Macdonald polynomials $I_{\lambda|N}$ was exhibited; we prefer to use the notation from \cite{Ol}. Explicitly, $D_N^0 = 1, D_N^1, \ldots, D_N^N$ are linear operators, in the algebra of symmetric polynomials on $x_1, \ldots, x_N$ with coefficients in $\Q(q, t)$, given by
\begin{equation*}
\sum_{k=0}^N{D_N^kz^k} := \frac{1}{\prod_{1\leq i<j\leq N}{(x_i - x_j)}}\det_{1\leq i, j\leq N}\left[ x_j^{N - i - 1}\left\{ (x_j t^{1 - N} - 1)t^{N - i} z T_{q, x_j} + (x_j + z) \right\} \right],
\end{equation*}
where $T_{q, x_j}$ are the $q$-shift operators.
They diagonalize the interpolation Macdonald polynomials $I_{\lambda|N}(x_1, \ldots, x_N; q, t)$, $\ell(\lambda) \leq N$, in fact,
\begin{equation*}
D_N^k I_{\lambda|N}(x_1, \ldots, x_N; q, t) = e_k(q^{\lambda_1}, q^{\lambda_2}t^{-1}, \ldots, q^{\lambda_N}t^{1-N})
I_{\lambda|N}(x_1, \ldots, x_N; q, t),
\end{equation*}
where $e_k(y_1, \ldots, y_N)$ is an elementary symmetric polynomial.

It is natural to ask if there is some kind of limit operator $D^k = \lim_{N\rightarrow\infty}D_N^k$ acting on symmetric functions that diagonalizes the interpolation Macdonald functions $I_{\lambda}(\cdot; q, t)$, and to ask for an explicit formula for such operator.
The analogous problems for Macdonald and Sekiguchi-Debiard operators were answered by Nazarov-Sklyanin in \cite{NS1, NS3}.
Their work is our starting point; we wanted to know if there is an extension of their formulas to the inhomogeneous setting of interpolation Macdonald operators.
In this paper, using ideas from Nazarov-Sklyanin, we answer the question for interpolation Macdonald operators in the affirmative.
By considering a simple renormalization of the operators $D_N^k$, to be denoted $A_N^k$, we can take a natural projective limit $A^k := \lim_{N\rightarrow\infty}{A_N^k}$ that is an operator acting on the ring of symmetric functions.
Explicitly, these operators can be described nicely by the action of the generating function
\begin{equation*}
A_{\infty}(u; q, t) := 1 + \frac{A^1}{1 - u} + \frac{A^2}{(1 - u)(1 - ut)} + \frac{A^3}{(1 - u)(1 - ut)(1 - ut^2)} + \dots
\end{equation*}
on the interpolation Macdonald functions:
\begin{equation*}
A_{\infty}(u; q, t) I_{\mu}(\cdot; q, t) = \prod_{i=1}^{\infty}\left\{ \frac{q^{\mu_i} - t^{i-1}u}{1 - t^{i-1}u} \right\}\cdot I_{\mu}(\cdot; q, t),
\end{equation*}
for all partitions $\mu = (\mu_1 \geq \mu_2 \geq \dots)$. The definition of the operators $\{A_N^k\}_k$ in terms of $\{D_N^k\}_k$, as well as the formal definition of the operators $A^k$ as limits of the finite ones $A_N^k$ is made explicit in the text below, see Section $\ref{sec:operators}$.

The main result of this article is an explicit formula for the operators $\{A^k\}_{k\in\N}$.
The answer is given in terms of the dual Hall-Littlewood functions $\QQ_{\lambda}(\cdot; t)$ and a new family of inhomogeneous Hall-Littlewood functions, that we denote by $\F_{\lambda}(\cdot; t)$. The formula is the following
\begin{equation*}
A^k := \sum_{\ell(\lambda) = k}{t^{\lambda_1 + \lambda_2 + \dots}\F_{\lambda}(\cdot; t^{-1})(\QQ_{\lambda}(\cdot; t^{-1}))^*},
\end{equation*}
where $(\QQ_{\lambda}(\cdot; t^{-1}))^*$ is the adjoint of the operator of multiplication by $\QQ_{\lambda}(\cdot; t^{-1})$, with respect to the Macdonald inner product, see Theorem $\ref{thm:main}$ in the text.

The symmetric function $\F_{\lambda}(\cdot; t)$ is a limit of a sequence of (a new family of) inhomogeneous Hall-Littlewood polynomials $\F_{\lambda}(x_1, \ldots, x_N; t)$.
These polynomials are defined by a formula involving a sum over the symmetric group, and resemble very much the usual Hall-Littlewood polynomials.

We also study these polynomials in more depth and prove that they are degenerations of interpolation Macdonald polynomials:
\begin{equation*}
\F_{\lambda}(x_1, \ldots, x_N; t) = \lim_{q\rightarrow 0} I_{\lambda | N}(x_1, \ldots, x_N; 1/q, 1/t).
\end{equation*}
Thus one can derive further properties for our polynomials $\F_{\lambda}(x_1, \ldots, x_N; t)$, and for the corresponding functions $\F_{\lambda}(\cdot; t)$, as degenerations of known results for interpolation Macdonald polynomials.
As examples, we derive closed-form formulas for the generating functions of one-row and one-column inhomogeneous Hall-Littlewood functions. As a consequence of the closed expression for the generating function of one-row inhomogeneous Hall-Littlewood functions, we obtain the following vertex operator form for $A^1$:
\begin{equation*}
\begin{aligned}
A^1 =& \frac{t}{t - 1}\oint_{|z| \ll 1}
\frac{dz}{2\pi\sqrt{-1} z}
\exp\left(\sum_{n=1}^{\infty}  \frac{z^n(1 - t^{-n})}{n}p_n  \right)
\exp\left( \sum_{n=1}^{\infty}{ z^{-n}(q^n - 1)\frac{\partial}{\partial p_n} } \right)\\
&- \frac{t}{t - 1} + \frac{1 - q}{1 - t}\frac{\partial}{\partial p_1},
\end{aligned}
\end{equation*}
where $p_1, p_2, \ldots$ are the Newton power sums, also known as the \textit{collective variables} in the physics literature.

The Macdonald polynomial $M_{\lambda|N}(x_1, \ldots, x_N; q, t)$ is the top-degree homogeneous component of the interpolation Macdonald polynomial $I_{\lambda|N}(x_1, \ldots, x_N; q, t)$, i.e.,
\begin{equation}\label{eqn:intermacdonalds}
M_{\lambda|N}(x_1, \ldots, x_N; q, t) = \lim_{a\rightarrow +\infty}{\frac{I_{\lambda|N}(ax_1, \ldots, ax_N; q, t)}{a^{\deg I_{\lambda | N}}}}.
\end{equation}
As mentioned above, Nazarov-Sklyanin studied the same problem we are dealing with, but for Macdonald and Sekiguchi-Debiard operators, degenerations corresponding to Macdonald and Jack functions, see \cite{NS2, NS1, NS3, NS4, NS5}.
From the relation $(\ref{eqn:intermacdonalds})$ between the Macdonald and interpolation Macdonald polynomials, we can deduce the main theorem of \cite{NS3} (as well as its degeneration in \cite{NS1}) from Theorem $\ref{thm:main}$.
At the combinatorial level, our main theorem is a \textit{refined Cauchy identity} for interpolation Macdonald functions; the arguments to prove this combinatorial identity differ from those in the articles of Nazarov-Sklyanin.
Lastly, let us mention that the operators considered in \cite{NS2, NS4, NS5}, form a \textit{different} hierarchy of operators than the ones in papers \cite{NS1, NS3}, and they admit an explicit formula via a Lax matrix. It would be interesting to see if one can apply similar techniques to produce such a different hierarchy of operators for interpolation Macdonald functions.

We end this introduction with the organization of this article. In Section $\ref{sec:symmetricfns}$, we recall some basics of symmetric functions from \cite{M}, and some facts about interpolation functions from \cite{Ol}. The new material is Subsection $\ref{sec:inhHL}$, where we construct an inhomogeneous analogue to the Hall-Littlewood function. Next, in Section $\ref{sec:operators}$, we discuss the operators of Okounkov that diagonalize the interpolation Macdonald polynomials, define the limit of such operator and state the main theorem.
The proof of the main theorem is based on the \textit{symbol} of our operators with respect to the Macdonald inner product, and it reduces to a combinatorial identity, like in \cite{NS1, NS3}. However, our proof of the resulting identity differs from these papers, and is carried out in Section $\ref{sec:mainproof}$. As a complement to the main result, in Section $\ref{sec:combinatorialprops}$ we study the inhomogeneous Hall-Littlewood polynomials; our second main result states that they are limits of interpolation Macdonald polynomials as $q\rightarrow 0$. As an application, we also find a vertex operator representation of the first operator $A^1$ of the hierarchy.

\subsection*{Acknowledgements}

The author is grateful to Alexei Borodin, and especially to Grigori Olshanski for helpful conversations.
In particular, Grigori Olshanski conjectured Theorem $\ref{cor:equalityHLs}$; sections $\ref{sec:Hall1}$ and $\ref{sec:Hall2}$ of the text are largely based on my discussions with him.

\section{Symmetric Functions}\label{sec:symmetricfns}

\subsection{Preliminaries}

We assume that the reader is familiar with the language of symmetric polynomials and symmetric functions, as in \cite[Ch. I]{M}.
We recall only a few notions and set our terminology.

For any $N\in\N$, let $\Lambda_{N, \FF}$ be the algebra of symmetric polynomials on $N$ variables $x_1, \ldots, x_N$ and with coefficients in the field $\FF$ (also $\Lambda_{0, \FF} := \FF$).
In what follows we only need $\FF = \Q$, $\Q(t)$, or $\Q(q, t)$, for two formal parameters $q, t$.
The substitution $x_N = 0$ gives a homomorphism $\pi_N: \Lambda_{N, \FF} \rightarrow \Lambda_{N-1, \FF}$ for any $N\geq 1$.
The projective limit of the chain $\{\Lambda_{N, \FF}, \pi_N\}$, in the category of graded algebras, is denoted $\Lambda_{\FF}$ and called the algebra of symmetric functions.
Equivalently, $\Lambda_{\FF}$ is the polynomial ring on infinitely many variables $p_1, p_2, \ldots$ that are identified with the Newton power sums.

For any $N\in\N$, the set of partitions $\lambda$ of length $\ell(\lambda) \leq N$ is denoted $\Y(N)$.
We also let $\Y(0)$ be the singleton containing only the empty partition, and $\Y := \cup_{N\geq 0}{\Y(N)}$ the set of all partitions.
All familiar bases of $\Lambda_{N, \FF}$ are parametrized by $\Y(N)$, for example the set $\{m_{\lambda}(x_1, \ldots, x_N) : \lambda\in\Y(N)\}$ of monomial symmetric polynomials,
\begin{equation}\label{eqn:monomial}
m_{\lambda}(x_1, \ldots, x_N) := \sum_{1\leq i_1 < \ldots < i_k \leq N}\sum_{\sigma\in S_{k}} c_{\lambda}^{-1}x_{i_{\sigma(1)}}^{\lambda_1}\cdots x_{i_{\sigma(k)}}^{\lambda_k},
\end{equation}
is a basis of $\Lambda_{N, \Q}$. In the display $(\ref{eqn:monomial})$, we used $k$ instead of $\ell(\lambda)$, and $c_{\lambda} := m_1(\lambda)!\cdot m_2(\lambda)!\cdots$, where $m_i = m_{i}(\lambda)$ denotes the multiplicity of $i$ in the partition $\lambda$. With this notation, we usually write $\lambda = (1^{m_1}2^{m_2}\cdots)$.

The monomial symmetric polynomials are stable: $m_{\lambda}(x_1, \ldots, x_N, 0) = m_{\lambda}(x_1, \ldots, x_N)$.
This property leads us to consider the monomial symmetric function $m_{\lambda} = m_{\lambda}(x_1, x_2, \ldots)$ in $\Lambda_{\Q}$.
Clearly $\{m_{\lambda} : \lambda\in\Y\}$ is a basis of $\Lambda_{\Q}$.
A different basis of $\Lambda_{\Q}$ is given by $\{p_{\mu} := \prod_{i\geq 1}{p_{\mu_i}} \ \forall \mu\in\Y\}$.
In the next subsections, we recall other bases of $\Lambda_{\FF}$, for $\FF = \Q(t)$ and $\Q(q, t)$, including those given by Hall-Littlewood and Macdonald functions.

\subsection{Hall-Littlewood functions}

Let $N\in\N$, $\lambda\in\Y(N)$.
The Hall-Littlewood (to be abbreviated HL henceforth) polynomial $\PP_{\lambda}(x_1, \ldots, x_N; t) \in \Lambda_{N, \Q(t)}$ is defined by, see \cite[Ch. III]{M},
\begin{equation}\label{eqn:HL}
\PP_{\lambda}(x_1, \dots, x_N; t) := \prod_{i \geq 0}\prod_{j = 1}^{m_i}{\frac{1 - t}{1 - t^j}}
\sum_{w\in S_N}{w\left\{ \prod_{i=1}^N{x_i^{\lambda_i}}  \prod_{1\leq i < j \leq N}{\frac{x_i - tx_j}{x_i - x_j}} \right\}},
\end{equation}
where $wf(x_1, \ldots, x_N) = f(x_{w(1)}, \ldots, x_{w(N)})$ for any permutation $w\in S_N$ and any function $f(x_1, \ldots, x_N)$, $\lambda = (1^{m_1}2^{m_2}\dots)$, and $m_0 = m_0(\lambda) = N - \ell(\lambda)$. If $\ell(\lambda) > N$, we set $\PP_{\lambda}(x_1, \ldots, x_N; t) := 0$.

It is known that $\PP_{\lambda}(x_1, \ldots, x_N; t)$ is a polynomial in the variables $x_1, \ldots, x_N$, it has degree $|\lambda| := \sum_i{\lambda_i}$, and has coefficients in $\Z[t]$.
It is also known that $\PP_{\emptyset}(x_1, \ldots, x_N; t) = 1$, which implies
\begin{equation}\label{eqn:combinatorial}
\sum_{w\in S_N}{w\left\{ \prod_{1\leq i<j\leq N} \frac{x_i - tx_j}{x_i - x_j} \right\}} = \prod_{i=1}^N{\frac{1 - t^i}{1 - t}}.
\end{equation}
The HL polynomials also have the stability property
\begin{equation}\label{eqn:Hallstability}
\PP_{\lambda}(x_1, \ldots, x_N, 0; t) =
	\begin{cases}
            \PP_{\lambda}(x_1, \ldots, x_N; t), & \text{ if } \lambda\in\Y(N),\\
            0, & \text{ if } \lambda\notin\Y(N).
        \end{cases}
\end{equation}
This stability allows one to define the HL function $\PP_{\lambda}(\cdot; t) = \PP_{\lambda}(x_1, x_2, \ldots; t) \in \Lambda_{\Q(t)}$.
The set $\{\PP_{\lambda}(\cdot; t) : \lambda\in\Y\}$ is a basis of $\Lambda_{\Q(t)}$.
We shall also need the dual basis given by
\begin{equation}\label{eqn:dualbasis}
\QQ_{\lambda}(\cdot; t) := b_{\lambda}(t)\PP_{\lambda}(\cdot; t), \ \forall \lambda\in\Y, \ \textrm{ where } b_{\lambda}(t) := \prod_{i\geq 1}\prod_{j=1}^{m_i}{(1 - t^j)}.
\end{equation}
The basis $\{\QQ_{\lambda}(\cdot; t) : \lambda\in\Y \}$ is dual to $\{\PP_{\lambda}(\cdot; t) : \lambda\in\Y\}$ with respect to the inner product in $(\ref{eqn:macdonaldinner})$ with $q = 0$; we shall not make use of this fact.
This duality is equivalent to the following combinatorial identity, known as the \textit{Cauchy identity for HL polynomials}:
\begin{equation}\label{eqn:cauchyHL}
\sum_{\lambda \in \Y}{P_{\lambda}(x_1, x_2, \ldots; t)Q_{\lambda}(y_1, y_2, \ldots; t)} = \prod_{i \geq 1}\prod_{j \geq 1}{ \frac{1 - tx_iy_j}{1 - x_iy_j} }.
\end{equation}

\subsection{Macdonald functions}

The Macdonald polynomials $M_{\lambda|N}(x_1, \ldots, x_N; q, t) \in \Lambda_{N, \Q(q, t)}$, $\lambda\in\Y(N)$, are a generalization of HL polynomials. To define them, introduce the \textit{Macdonald difference operators}
\begin{equation}\label{eqn:macdonaldops}
H_N^0 := 1, \hspace{.2in} H_N^r := t^{{r \choose 2}}\sum_{\substack{ I \subset \{1, \ldots, N\} \\ |I| = r }}
\prod_{\substack{i\in I \\ j\notin I}}{\frac{tx_i - x_j}{x_i - x_j}}
\prod_{i\in I}{T_{q, x_i}} \ \forall 1\leq r\leq N.
\end{equation}
For any $\lambda\in\Y(N)$, the Macdonald polynomial $M_{\lambda|N}(x_1, \ldots, x_N; q, t)$ is the unique element of $\Lambda_{N, \Q(q, t)}$ of the form $m_{\lambda}(x_1, \ldots, x_N) + \textrm{ terms }c_{\mu}m_{\mu}(x_1, \ldots, x_N)$, where $\mu$ ranges over partitions in $\Y(N)$ with $|\mu| = |\lambda|$ and $\mu < \lambda$ in the dominance order, and moreover
\begin{equation*}
H_N^r M_{\lambda|N}(x_1, \ldots, x_N; q, t) = e_r(q^{\lambda_1}t^{N-1}, \ldots, q^{\lambda_N})\cdot M_{\lambda|N}(x_1, \ldots, x_N; q, t) \ \forall r = 0, 1, \ldots, N,
\end{equation*}
where $e_r(y_1, \ldots, y_N) = \sum_{1\leq i_1 < \ldots < i_r \leq N}{y_{i_1}\cdots y_{i_r}}$ is an elementary symmetric polynomial.

For each $N\in\N$, the set $\{M_{\lambda|N}(x_1, \ldots, x_N; q, t) : \lambda\in\Y(N)\}$ is a basis of $\Lambda_{N, \Q(q, t)}$.
The Macdonald polynomials are stable, i.e., $M_{\lambda|N+1}(x_1, \ldots, x_N, 0; q, t) = M_{\lambda|N}(x_1, \ldots, x_N; q, t)$, which allows us to define the Macdonald functions $M_{\lambda}(\cdot; q, t)\in\Lambda_{\Q(q, t)}$.
It follows that $\{M_{\lambda}(\cdot; q, t) : \lambda\in\Y\}$ is a basis of $\Lambda_{\Q(q, t)}$; it is known that this basis is orthogonal with respect to the inner product $(\cdot, \cdot)_{q, t}$, defined by declaring
\begin{equation}\label{eqn:macdonaldinner}
(p_{\mu}, p_{\nu})_{q, t} := \delta_{\mu, \nu}\prod_{i\geq 1}{(i^{m_i(\mu)}m_i(\mu)!)}\cdot\prod_{i=1}^{\ell(\mu)}{\frac{1 - q^{\mu_i}}{1 - t^{\mu_i}}} \ \forall \mu, \nu\in\Y.
\end{equation}
Macdonald functions generalize HL functions, since $\Lambda_{\Q(q, t)} \xrightarrow{q = 0} \Lambda_{\Q(t)}$ maps $M_{\lambda}(\cdot; q, t)$ to $\PP_{\lambda}(\cdot; t)$.

\subsection{Interpolation Macdonald functions}\label{sec:interpolationpolys}

We define the interpolation Macdonald functions in the notation of Olshanski \cite{Ol2, Ol}.
They also appear in papers by Knop \cite{K}, Okounkov \cite{Ok3, Ok, Ok2}, and Sahi \cite{S}\footnote{Actually, only the interpolation Macdonald \textit{polynomials} appear in the papers \cite{K, Ok3, Ok, Ok2, S}. It seems like Olshanski was the first who considered their lifts to the ring of symmetric functions. However, Sergeev and Veselov \cite{SV}, as well as Rains \cite{R}, have previously studied lifts of \textit{Jack-Laurent polynomials} and \textit{$BC_N$-symmetric polynomials}, respectively}.
Unlike the functions from previous subsections, the interpolation Macdonald functions are inhomogeneous.

For any $\lambda\in\Y(N)$, the \textit{interpolation Macdonald polynomial} $I_{\lambda|N}(x_1, \ldots, x_N; q, t)$ is the unique element of $\Lambda_{N, \Q(q, t)}$ satisfying the following conditions:
\begin{enumerate}
	\item $\deg I_{\lambda|N}(x_1, \ldots, x_N; q, t) = |\lambda|$,
	\item $I_{\lambda|N}(q^{-\mu_1}, q^{-\mu_2}t, \ldots, q^{-\mu_N}t^{N-1}; q, t) = 0$ for all $\mu \neq \lambda$ with $|\mu| \leq |\lambda|$,
	\item $I_{\lambda|N}(q^{-\lambda_1}, q^{-\lambda_2}t, \ldots, q^{-\lambda_N}t^{N-1}; q, t) \neq 0$,
	\item $I_{\lambda|N}(x_1, \ldots, x_N; q, t) = m_{\lambda}(x_1, \cdots, x_N) + \cdots$,
\end{enumerate}
where the dots in the last condition stand for a linear combination of monomial symmetric polynomials $m_{\mu}(x_1, \cdots, x_N)$, with either $|\mu| < |\lambda|$, or $|\mu| = |\lambda|$ and $\mu \leq \lambda$ in the dominance order of partitions.

Items 3 and 4 of the definition can be replaced by
\begin{equation}\label{eqn:alternative}
I_{\lambda | N}(q^{-\lambda_1}, q^{-\lambda_2}t, \ldots, q^{-\lambda_N}t^{N-1}; q, t) = C(\lambda; q, t),
\end{equation}
where the normalization constant $C(\lambda; q, t)$ is
\begin{equation}\label{Hexp}
C(\lambda; q, t) := q^{-2n(\lambda') - |\lambda|}t^{n(\lambda)}\prod_{s\in\lambda}{(1 - q^{a(s)+1}t^{l(s)})}.
\end{equation}
In formula $(\ref{Hexp})$, we need to recall some notation: for any partition $\kappa = (\kappa_1, \kappa_2, \dots)$, we denote
\begin{equation*}
n(\kappa) := \sum_{i \geq 1}{(i-1)\kappa_i} = \sum_{j \geq 1}{\kappa_j' \choose 2},
\end{equation*}
and for any square in the Young diagram $s = (i, j) \in \kappa$, we let
\begin{gather*}
a(s) := \lambda_i - j, \hspace{.2in} l(s) := \lambda_j' - i,\\
a'(s) := j-1, \hspace{.2in} l'(s) = i - 1,
\end{gather*}
be the arm length, leg length, coarm length and coleg length of $s = (i, j)$, respectively.

It turns out that the top-degree homogeneous component of $I_{\lambda|N}(x_1, \ldots, x_N; q, t)$ is the Macdonald polynomial $M_{\lambda|N}(x_1, \ldots, x_N; q, t)$.
As a consequence, $\{ I_{\lambda|N}(x_1, \ldots, x_N; q, t) : \lambda\in\Y(N) \}$ is a basis of $\Lambda_{N, \Q(q, t)}$.

The relation between Okounkov's notation $P_{\lambda}^*(x_1, \ldots, x_N; q, t)$ and ours is
\begin{equation*}
I_{\lambda | N}(x_1, \ldots, x_N; q, t) = P_{\lambda}^*(x_1, x_2/t, \ldots, x_N/t^{N-1}; 1/q, 1/t).
\end{equation*}
Thus the translation of \cite[(1.9)]{Ok3} to our notation is the special value
\begin{equation*}
I_{\mu | N}(a, at, \ldots, at^{N-1}; q, t) = t^{n(\mu)}\cdot\prod_{s\in\mu}\frac{(1 - q^{a'(s)}t^{N-l'(s)})(a - q^{-a'(s)}t^{l'(s)})}{(1 - q^{a(s)}t^{l(s)+1})}.
\end{equation*}
In particular, by using
\begin{equation*}
\sum_{s\in\mu}{a'(s)} = n(\mu'), \hspace{.1in}\sum_{s\in\mu}{l'(s)} = n(\mu),
\end{equation*}
we obtain
\begin{equation}\label{evalzeros}
I_{\mu | N}(0^N; q, t) = I_{\mu | N}(\underbrace{0, \ldots, 0}_{N \textrm{ zeroes}}; q, t) = (-1)^{|\mu|} q^{-n(\mu')} t^{2n(\mu)}\cdot\prod_{s\in\mu}\frac{1 - q^{a'(s)}t^{N-l'(s)}}{1 - q^{a(s)}t^{l(s)+1}}.
\end{equation}

In the spirit of the classical theory of symmetric functions, we want to define an element of $\Lambda_{\Q(q, t)}$ that is a limit of the interpolation Macdonald polynomials $I_{\lambda|N}(x_1, \ldots, x_N; q, t)$, as $N$ tends to infinity. The suitable stability property is the following:
\begin{equation}\label{eqn:interpolstable}
I_{\lambda|N+1}(x_1, \ldots, x_N, t^N; q, t) = I_{\lambda|N}(x_1, \ldots, x_N; q, t).
\end{equation}

Now consider the chain of algebra homomorphisms
\begin{equation*}
\Q(q, t) \xleftarrow{x_1 = 1} \Lambda_{1, \Q(q, t)} \xleftarrow{x_2 = t} \Lambda_{2, \Q(q, t)} \leftarrow \cdots \leftarrow \Lambda_{N-1, \Q(q, t)} \xleftarrow{x_{N} = t^{N-1}} \Lambda_{N, \Q(q, t)} \xleftarrow{x_{N+1} = t^N}\cdots
\end{equation*}
and let $\Lambda_{\Q(q, t)}'$ be the projective limit in the category of filtered algebras.
Because of the stability property $(\ref{eqn:interpolstable})$, the sequence $\{I_{\lambda|N}(x_1, \ldots, x_N; q, t) : N \geq \ell(\lambda)\}$ defines an element of $\Lambda_{\Q(q, t)}'$, for each $\lambda\in\Y$.

A natural algebra isomorphism $\Lambda_{\Q(q, t)} \rightarrow \Lambda_{\Q(q, t)}'$ is defined by sending each Newton power sum $p_n$, $n\geq 1$, to some ``regularization'' of $\{p_n(x_1, \ldots, x_N) + t^{nN} + t^{n(N+1)} + \dots : N\geq 1\}$; formally, the map is the unique algebra homomorphism such that
\begin{equation*}
\Lambda_{\Q(q, t)} \ni p_n \mapsto \left\{p_n(x_1, \ldots, x_N) + \frac{t^{nN}}{1 - t^{n}} \right\} = \left\{x_1^n + \ldots + x_N^n + \frac{t^{nN}}{1 - t^{n}} \right\} \in \Lambda_{\Q(q, t)}'.
\end{equation*}
It is easily shown to be an isomorphism, see the proof of Proposition $\ref{df:Halls}$ for a similar statement.
The element of $\Lambda_{\Q(q, t)}$ that corresponds to $\{I_{\lambda|N}(x_1, \ldots, x_N; q, t) : N \geq \ell(\lambda)\} \in \Lambda_{\Q(q, t)}'$ is denoted by $I_{\lambda}(\cdot; q, t) = I_{\lambda}(x_1, x_2, \ldots; q, t)$ and called the \textit{interpolation Macdonald function parametrized by $\lambda\in\Y$}. The set $\{I_{\lambda}(\cdot; q, t) : \lambda\in\Y\}$ is an (inhomogeneous) basis of $\Lambda_{\Q(q, t)}$.

\subsection{Inhomogeneous Hall-Littlewood functions}\label{sec:inhHL}

To write down an expression for the limits of interpolation Macdonald operators, we need a new family of symmetric functions, which are a sort of \textit{inhomogeneous} version of HL functions. We study some of their properties, from a purely combinatorial point of view.

\begin{df}\label{def:inhHL}
Let $N\in\N$ and $\lambda\in\Y(N)$ a partition of length $0\leq \ell(\lambda) = k \leq N$.
Define the \textit{inhomogeneous Hall-Littlewood polynomial} $\F_{\lambda}(x_1, \ldots, x_N; t) \in \Lambda_{N, \Q(t)}$ by
\begin{equation}\label{def:inh1}
\F_{\lambda}(x_1, \dots, x_N; t) := \prod_{i \geq 0}\prod_{j = 1}^{m_i(\lambda)}{\frac{1 - t}{1 - t^j}}
\sum_{w\in S_N}{w\left\{ \prod_{i=1}^k{\left((1 - t^{1-N}x_i^{-1})x_i^{\lambda_i}\right)  \prod_{1\leq i < j \leq N}{\frac{x_i - tx_j}{x_i - x_j}}} \right\}},
\end{equation}
where $m_0(\lambda) = N - \ell(\lambda)$.
If $\lambda\in\Y\setminus\Y(N)$, i.e. if $\ell(\lambda) > N$, then set
\begin{equation}\label{def:inh2}
\F_{\lambda}(x_1, \ldots, x_N; t) := 0.
\end{equation}
\end{df}

\begin{rem}\label{rem:specialcases}
When $\ell(\lambda) = 0$, i.e., when $\lambda = \emptyset$, then $\F_{\emptyset}(x_1, \ldots, x_N; t) = 1$.
When $\ell(\lambda) = N$, then $(\ref{eqn:HL})$ and $(\ref{def:inh1})$ imply $\F_{\lambda}(x_1, \ldots, x_N; t) = \prod_{i=1}^N{(1 - t^{1-N}x_i^{-1})}\PP_{\lambda}(x_1, \ldots, x_N; t)$.
\end{rem}

\begin{rem}\label{rem:borodin}
Our symmetric polynomials $\F_{\lambda}(x_1, \ldots, x_N; t)$ look very similar to the symmetric rational functions $G_{\lambda}(u_1, \ldots, u_N; \Xi, S)$ of Borodin-Petrov \cite{BP}, e.g., compare $(\ref{def:inhHL})$ with \cite[(4.24)]{BP}. However, our polynomials seem not to be a degeneration of the family of functions of Borodin-Petrov.
\end{rem}

\begin{rem}
Borodin defined a different kind of inhomogeneous HL polynomial as a degeneration of the homogeneous version of the symmetric rational function $G_{\lambda}(u_1, \ldots, u_N; \Xi, S)$; see \cite[Sec. 8.2]{B} for their definition.
\end{rem}

It is not immediately clear from the definition, but $\F_{\lambda}(x_1, \ldots, x_N; t)$ is a symmetric polynomial in the variables $x_1, \ldots, x_N$, of degree $|\lambda|$, and with coefficients in $\Z[t, t^{-1}]$.
This can be proved in the same fashion as one proves the analogous properties for HL polynomials, as defined in $(\ref{eqn:HL})$, see \cite[Ch. III, Sec. 1]{M} for details.

From $(\ref{eqn:HL})$ and $(\ref{def:inh1})$, we deduce that if $\ell(\lambda) \leq N$, the top-degree homogeneous component of $\F_{\lambda}(x_1, \ldots, x_N; t)$ is the HL polynomial $\PP_{\lambda}(x_1, \ldots, x_N; t)$.
It follows that $\{\F_{\lambda}(x_1, \ldots, x_N; t) : \lambda\in\Y(N)\}$ is a basis of $\Lambda_{N, \Q(t)}$.

\begin{lem}\label{lem:rewrite}
Let $N\in\N$ and $\lambda\in\Y(N)$ be a partition of length $0\leq \ell(\lambda) = k \leq N$. Then
\begin{equation*}
\F_{\lambda}(x_1, \dots, x_N; t) = \sum_{w\in S_N/S_N^{\lambda}}{w\left\{ \prod_{i=1}^k{\left((1 - t^{1-N}x_i^{-1})x_i^{\lambda_i}\right)  \prod_{\substack{1\leq i < j \leq N \\ \lambda_i > \lambda_j}}{\frac{x_i - tx_j}{x_i - x_j}}} \right\}},
\end{equation*}
where $S_{N}^{\lambda} := \{ w\in S_N : \lambda_{w(i)} = \lambda_i \ \forall i = 1, 2, \ldots, N\}$ and the sum above is over representatives $w$ of the coset space $S_N/S_N^{\lambda}$.
\end{lem}
\begin{proof}
The proof follows exactly the argument of \cite[Ch. III, (1.5)]{M}. Details are left to the reader.
\end{proof}

\begin{prop}\label{prop:coherence}
Let $N\in\N$, $N\geq 2$, and $\lambda\in\Y$. Then
\begin{equation}\label{eqn:coherence}
\F_{\lambda}(x_1, \ldots, x_{N-1}, t^{1-N}; t) = \F_{\lambda}(x_1, \ldots, x_{N-1}; t)
\end{equation}
\end{prop}
\begin{proof}
If $\ell(\lambda) > N$, then both sides of the equality are zero, by definition $(\ref{def:inh2})$.

If $\ell(\lambda) = N$, then the right side of the equality is zero by $(\ref{def:inh2})$. On the other hand, since $\ell(\lambda) = N$, then the product over $i$ inside the brackets of $(\ref{def:inh1})$ is invariant under permutations of its variables by $w\in S_N$, so that product can be factored out of the sum. It follows that
\begin{equation*}
\F_{\lambda}(x_1, \dots, x_N; t) = \prod_{i \geq 0}\prod_{j = 1}^{m_i(\lambda)}{\frac{1 - t}{1 - t^j}}
\prod_{i=1}^N{\left((1 - t^{1-N}x_i^{-1})x_i^{\lambda_i}\right)}
\sum_{w\in S_N}{w\left\{  \prod_{1\leq i < j \leq N}{\frac{x_i - tx_j}{x_i - x_j}} \right\}}.
\end{equation*}
It is clear from the formula above that the factor $1 - t^{1-N}x_N^{-1}$ vanishes if $x_N = t^{1-N}$, and therefore $\F_{\lambda}(x_1, \ldots, x_{N-1}, t^{1-N}; t) = 0$.

Finally assume $0\leq \ell(\lambda) = k \leq N - 1$. In this case, we can apply Lemma $\ref{lem:rewrite}$ to both sides of $(\ref{eqn:coherence})$. When we apply it to the left side, we note that some terms in the sum over $w\in S_N/S_N^{\lambda}$ are zero.
In fact, those terms corresponding to $w\in S_N/S_N^{\lambda}$ such that $w^{-1}(N) \in \{1, 2, \ldots, k\}$ vanish because the factor $\prod_{i=1}^k{((1 - t^{1-N}x_{w(i)}^{-1})x_{w(i)}^{\lambda_i})}$ equals zero when $x_N = t^{1-N}$.

Thus the sum is actually over those $w\in S_N/S_N^{\lambda}$ with $w^{-1}(N)\in\{k+1, \ldots, N\}$.
Note that permutations of $\{k+1, \ldots, N\}$ belong to $S_N^{\lambda}$ since $\lambda_{k+1} = \ldots = \lambda_N = 0$.
This means that each coset representative $w\in S_N/S_N^{\lambda}$ can be chosen so that $w(N) = N$, and therefore the restriction $w|_{S_{N-1}}$ runs naturally over elements of $S_{N-1}/S_{N-1}^{\lambda}$, as $w$ runs over $S_N/S_N^{\lambda}$.
Finally observe that $w(N) = N$ implies
\begin{equation*}
\begin{gathered}
\left.\prod_{i=1}^k{\left((1 - t^{1-N}x_{w(i)}^{-1})x_{w(i)}^{\lambda_i}\right)}
\prod_{\substack{1\leq i < j \leq N \\ \lambda_i > \lambda_j}}{\frac{x_{w(i)} - tx_{w(j)}}{x_{w(i)} - x_{w(j)}}}\right|_{x_N = t^{1-N}}\\
= \prod_{i=1}^k{\left((1 - t^{1-N}x_{w(i)}^{-1})x_{w(i)}^{\lambda_i}\right)}  \prod_{\substack{1\leq i < j \leq N-1 \\ \lambda_i > \lambda_j}}{\frac{x_{w(i)} - tx_{w(j)}}{x_{w(i)} - x_{w(j)}}}
\prod_{i=1}^{k}{\frac{x_{w(i)} - t^{2-N}}{x_{w(i)} - t^{1-N}}}\\
= \prod_{i=1}^k{\left((1 - t^{2-N}x_{w(i)}^{-1})x_{w(i)}^{\lambda_i}\right)}\prod_{\substack{1\leq i < j \leq N-1 \\ \lambda_i > \lambda_j}}{\frac{x_{w(i)} - tx_{w(j)}}{x_{w(i)} - x_{w(j)}}}
\end{gathered}
\end{equation*}
and we end up with the formula of Lemma $\ref{lem:rewrite}$ for the right side of $(\ref{eqn:coherence})$.
\end{proof}

Using Proposition $\ref{prop:coherence}$, we construct for each $\lambda\in\Y$ an element of $\Lambda_{\Q(t)}$ that uniquely corresponds to the coherent sequence $\{ \F_{\lambda}(x_1, \ldots, x_N; t) : N \geq \ell(\lambda) \}$.
For any $N\in\N$, consider the map
\begin{equation}\label{eqn:mapsN}
\pi^N_{N-1}: \Lambda_{N, \Q(t)} \xrightarrow{x_N = t^{1-N}} \Lambda_{N-1, \Q(t)}
\end{equation}
given by the specialization $x_N = t^{1-N}$.
Also, for any $N\in\N$, consider the following (unital) algebra homomorphism specified by the action on the set of Newton power sums $\{p_m : m \geq 1\}$, which is a  generator set for $\Lambda_{\Q(t)}$:
\begin{equation}\label{eqn:mapsinfty}
\begin{gathered}
\pi^{\infty}_N : \Lambda_{\Q(t)} \longrightarrow \Lambda_{N, \Q(t)},\\
\pi^{\infty}_N(p_m) = p_m(x_1, \ldots, x_N) + \frac{t^{-mN}}{1 - t^{-m}}, \textrm{ for all } m \geq 1,
\end{gathered}
\end{equation}
where $p_m(x_1, \ldots, x_N) = x_1^m + \ldots + x_N^m$.
The expression $t^{-mN}/(1 - t^{-m})$ in $(\ref{eqn:mapsinfty})$ comes from the geometric sum $(t^{-N})^m + (t^{-N-1})^m + \ldots$, if we assume that $t > 1$.

\begin{dfprop}\label{df:Halls}
Let $\lambda\in\Y$ be arbitrary. There exists a unique $\F_{\lambda}(\cdot; t) \in \Lambda_{\Q(t)}$ such that $\pi^{\infty}_N \F_{\lambda}(\cdot; t) = \F_{\lambda}(x_1, \ldots, x_N; t)$ for all $N\in\N$.
We call such unique element the \textit{inhomogeneous Hall-Littlewood function} $\F_{\lambda}(\cdot; t)$ \textit{parametrized by} $\lambda$. The set $\{ \F_{\lambda}(\cdot; t) : \lambda\in\Y \}$ is a basis of $\Lambda_{\Q(t)}$.
\end{dfprop}
\begin{proof}
The proof is similar to that of \cite[Prop. 2.8]{Ol}. We repeat the proof, with necessary modifications, for the reader's convenience.

Let $\Lambda_{\Q(t)}'$ be the projective limit of the chain $\{\Lambda_{N, \Q(t)}, \pi^{N+1}_N\}$ in the category of filtered algebras, and for each $N\in\N$, let $(\pi^{\infty}_N)': \Lambda_{\Q(t)}' \rightarrow \Lambda_{N, \Q(t)}$ be the projection.

Next, define a homomorphism $\pi: \Lambda_{\Q(t)} \rightarrow \Lambda_{\Q(t)}'$ specified by the action on the set of Newton power sums $\{p_n : n \geq 1\}$, as follows:
\begin{equation*}
\Lambda_{\Q(t)} \ni p_m \mapsto \left\{p_m(x_1, \ldots, x_N) + \frac{t^{-mN}}{1 - t^{-m}} \right\} = \left\{x_1^m + \ldots + x_N^m + \frac{t^{-mN}}{1 - t^{-m}} \right\} \in \Lambda_{\Q(t)}'.
\end{equation*}
Note that the sequence above determines an element of $\Lambda_{\Q(t)}'$ because of the coherence relation
\begin{equation*}
(x_1^m + \ldots + x_{N-1}^m + (t^{1-N})^m) + \frac{t^{-mN}}{1 - t^{-m}} = (x_1^m + \ldots + x_{N-1}^m) + \frac{t^{-m(N-1)}}{1 - t^{-m}}.
\end{equation*}

Observe that for any fixed $d\in\N$, the map $\pi^{N+1}_N$ induces a linear isomorphism between the subspaces of degree $\leq d$ in $\Lambda_{N, \Q(t)}$ and $\Lambda_{N-1, \Q(t)}$, provided that $N > d$.
If follows that $\pi: \Lambda_{\Q(t)} \rightarrow \Lambda_{\Q(t)}'$ is an algebra isomorphism.
Lastly note that $(\pi_N^{\infty})'\circ\pi : \Lambda_{\Q(t)} \rightarrow \Lambda_{N, \Q(t)}$ maps each $p_m$ to $p_m(x_1, \ldots, x_N) + \frac{t^{-mN}}{1 - t^{-m}}$, meaning that the composition coincides with the map $\pi^{\infty}_N$ in $(\ref{eqn:mapsinfty})$.

Due to the stability property of Proposition $\ref{prop:coherence}$, the sequence $\{F_{\lambda|N}(x_1, \ldots, x_N; t)\}_N$ defines an element of $\Lambda_{\Q(t)}'$.
Let $F_{\lambda}(\cdot; q, t)$ be the corresponding element of $\Lambda_{\Q(t)}$. From the discussion above, it is clear that it has the required property. The uniqueness is a consequence of the fact that $\pi: \Lambda_{\Q(t)} \rightarrow \Lambda_{\Q(t)}'$ is an isomorphism.

Finally, recall that for each $N\in\N$ the top homogeneous part of $\F_{\lambda}(x_1, \ldots, x_N; t)$ is the HL polynomial $\PP_{\lambda}(x_1, \ldots, x_N; t)$; by the construction it follows that the top-degree homogeneous component of $\F_{\lambda}(\cdot; t)$ is the HL function $\PP_{\lambda}(\cdot; t)$. Since $\{\PP_{\lambda}(\cdot; t) : \lambda\in\Y\}$ is a basis of $\Lambda_{\Q(t)}$, then so is the set $\{ \F_{\lambda}(\cdot; t) : \lambda\in\Y \}$, thus the last statement is proved.
\end{proof}

\begin{rem}\label{rem:other}
Instead of $\pi^N_{N-1}$ and $\pi^{\infty}_N$, consider the maps $\widetilde{\pi}^N_{N-1} : \Lambda_{N, \Q(t)} \rightarrow \Lambda_{N-1, \Q(t)}$, $x_N = t^{N-1}$, and $\widetilde{\pi}^{\infty}_N : \Lambda_{\Q(t)} \rightarrow \Lambda_{N, \Q(t)}$, $p_m \mapsto p_m(x_1, \ldots, x_N) + t^{mN}/(1 - t^m)$ for all $m \geq 1$.
A similar construction as that of Proposition/Definition $\ref{df:Halls}$ gives a new family of symmetric functions that we naturally denote by $\F_{\lambda}(\cdot; t^{-1})$.
\end{rem}

\section{Interpolation Macdonald operators and their limits}\label{sec:operators}

\subsection{Operators for interpolation Macdonald polynomials}

We recall the operators of Okounkov \cite{Ok} in the notation of \cite{Ol}.
The \textit{interpolation Macdonald operators} $D_N^1, D_N^2, \ldots, D_N^N$ on $\Lambda_{N, \Q(q, t)}$ are defined by the equations
\begin{equation}\label{eqn:Dop}
\begin{gathered}
D_N(z; q, t)  := 1 + \sum_{k=1}^N{D_N^kz^k},\\
D_N(z; q, t) := \frac{1}{V(x_1, \ldots, x_N)}\det_{1\leq i, j\leq N}\left[ x_j^{N - i - 1}\left\{ (x_j t^{1 - N} - 1)t^{N - i} z T_{q, x_j} + (x_j + z) \right\} \right],
\end{gathered}
\end{equation}
where $V(x_1, \ldots, x_N) := \prod_{1\leq i<j\leq N}{(x_i - x_j)}$ is the Vandermonde determinant, and $\{T_{q, x_j}\}_{1\leq j\leq N}$ are the $q$-shift operators, given by $(T_{q, x_j}f)(x_1, \ldots, x_N) := f(x_1, \ldots, x_{j-1}, qx_j, x_{j+1}, \ldots, x_N)$.
Observe that $D_N^1, \ldots, D_N^N$ depend on $q, t$, but we suppress them from the notation.
They diagonalize the interpolation Macdonald polynomials $\{I_{\mu|N}(x_1, \ldots, x_N; q, t) : \mu\in\Y(N)\}$; in fact,
\begin{equation}\label{eqn:diagonal}
D_N(z; q, t)I_{\mu|N}(x_1, \ldots, x_N; q, t) = \prod_{i=1}^N{(1 + q^{\mu_i}t^{1 - i}z)}\cdot I_{\mu|N}(x_1, \ldots, x_N; q, t) \ \forall \mu \in \Y(N).
\end{equation}
In particular, $\{D_N^k : 1\leq k\leq N\}$ is a pairwise commuting family of $q$-difference operators.
We prefer to consider a renormalization of these operators, namely $A_N^1, A_N^2, \ldots, A_N^N$, given by
\begin{equation*}
1 + \sum_{k=1}^N{\frac{A_N^k}{(u; t)_k}} = A_N(u; q, t) := \frac{D_N(-u^{-1}; q, t)(-u)^N t^{\frac{N(N-1)}{2}}}{(u; t)_N},
\end{equation*}
where $(u; t)_k := \prod_{i=0}^{k-1}{(1 - t^{i - 1}u)}$ is the usual Pochhammer symbol.
From $(\ref{eqn:Dop})$, we deduce
\begin{equation}\label{eqn:Aop}
A_N(u; q, t) = \frac{1}{V(x_1, \ldots, x_N)\cdot (u; t)_N}\circ \det_{1\leq i, j\leq N}\left[ x_j^{N - i - 1}\left\{ (x_j - t^{N-1})T_{q, x_j} + t^{i-1}(1 - x_j u) \right\} \right].
\end{equation}
They also diagonalize each $I_{\mu|N}(x_1, \ldots, x_N; q, t)$; in fact, $(\ref{eqn:diagonal})$ yields
\begin{equation}\label{eqn:diagonal2}
A_N(u; q, t)I_{\mu|N}(x_1, \ldots, x_N; q, t) = \prod_{i=1}^N{\left\{  \frac{q^{\mu_i} - t^{i-1}u}{1 - t^{i-1}u} \right\}} I_{\mu|N}(x_1, \ldots, x_N; q, t).
\end{equation}
As before, $\{A_N^k : 1\leq k\leq N\}$ is a pairwise commuting family of $q$-difference operators.
By treating $u$ as a complex variable and expanding the fractions $1/(1 - t^{i - 1}u)$ near $u = 0$, we obtain, in the right side of $(\ref{eqn:diagonal2})$, an element of $\Lambda_{N, \Q(q, t)}[[u]]$.
Since $\{I_{\mu|N}(x_1, \ldots, x_N; q, t) : \mu\in\Y(N)\}$ is a basis of $\Lambda_{N, \Q(q, t)}$, it follows that
\begin{equation*}
A_N(u; q, t) : \Lambda_{N, \Q(q, t)} \rightarrow \Lambda_{N, \Q(q, t)}[[u]]
\end{equation*}
is a well-defined operator.

\subsection{Operators for interpolation Macdonald functions}

We normalized $D_N(z; q, t)$ into $A_N(u; q, t)$ to obtain equation $(\ref{eqn:diagonal2})$.
The key property of this eigenrelation is that the factors entering the eigenvalue, namely $(q^{\mu_i} - t^{i-1}u)/(1 - t^{i-1}u)$, equal $1$ when $\mu_i = 0$. We can deduce, by using also that $\{I_{\mu|N}(x_1, \ldots, x_N; q, t) : \mu\in\Y(N)\}$ is a basis of $\Lambda_{N, \Q(q, t)}$, the following coherence property
\begin{equation}\label{eqn:coherenceops}
A_{N-1}(u; q, t)\pi^{N}_{N-1} = \pi^{N}_{N-1}A_{N}(u; q, t) : \Lambda_{N, \Q(q, t)} \rightarrow \Lambda_{N-1, \Q(q, t)}[[u]].
\end{equation}

\begin{lem}\label{lem:opinfty}
There exists a unique linear operator $A_{\infty}(u; q, t) : \Lambda_{\Q(q, t)} \rightarrow \Lambda_{\Q(q, t)}[[u]]$ such that
\begin{equation}\label{eqn:intertwining}
A_N(u; q, t)\pi^{\infty}_N = \pi^{\infty}_N A_{\infty}(u; q, t): \Lambda_{\Q(q, t)} \rightarrow \Lambda_{N, \Q(q, t)},
\end{equation}
for all $N \geq 1$. It is given by
\begin{equation}\label{eqn:defAinfty}
\begin{gathered}
A_{\infty}(u; q, t) : \Lambda_{\Q(q, t)} \rightarrow \Lambda_{\Q(q, t)}[[u]]\\
I_{\mu}(\cdot; q, t) \mapsto \prod_{i=1}^{\infty}\left\{ \frac{q^{\mu_i} - t^{i-1}u}{1 - t^{i-1}u} \right\}I_{\mu}(\cdot; q, t) \ \forall \mu\in\Y.
\end{gathered}
\end{equation}
\end{lem}
\begin{proof}

Observe that for any $\mu\in\Y$, the product in the display $(\ref{eqn:defAinfty})$ is finite; in fact, the only terms unequal to $1$ are those ranging from $i = 1$ to $i = \ell(\mu)$. Since $\{I_{\mu}(\cdot; q, t) : \mu\in\Y\}$ is a basis of $\Lambda_{\Q(q, t)}$, the definition above completely determines the operator. Moreover, $(\ref{eqn:defAinfty})$ and the definition of the interpolation Macdonald function $I_{\mu}(\cdot; q, t)$ easily imply
\begin{equation*}
A_N(u; q, t)\pi^{\infty}_N I_{\mu}(\cdot; q, t) = \pi^{\infty}_N A_{\infty}(u; q, t) I_{\mu}(\cdot; q, t) \ \forall \mu\in\Y.
\end{equation*}
Again, since $\{I_{\mu}(\cdot; q, t) : \mu\in\Y\}$ is a basis of $\Lambda_{\Q(q, t)}$, then the equality of operators $(\ref{eqn:intertwining})$ ensues.
\end{proof}

The following explicit formula for $A_{\infty}(u; q, t)$ is the main result of this paper.

\begin{thm}\label{thm:main}
We can write
\begin{equation}\label{eqn:expansion}
A_{\infty}(u; q, t) = 1 + \frac{A^1}{(u; t)_1} + \frac{A^2}{(u; t)_2} + \dots,
\end{equation}
where $A^1, A^2, \dots : \Lambda_{\Q(q, t)} \rightarrow \Lambda_{\Q(q, t)}$ are given by
\begin{equation}\label{eqn:kexplicit}
A^k = \sum_{\ell(\lambda) = k}{t^{\lambda_1 + \lambda_2 + \dots}\F_{\lambda}(\cdot; t^{-1})(\QQ_{\lambda}(\cdot; t^{-1}))^*}.
\end{equation}
In the formula above, $(Q_{\lambda}(\cdot; t^{-1}))^*$ is the adjoint of the operator $\Lambda_{\Q(q, t)} \rightarrow \Lambda_{\Q(q, t)}$ of multiplication by $Q_{\lambda}(\cdot; t^{-1})$ with respect to the Macdonald inner product $(\cdot, \cdot)_{q, t}$, and $F_{\lambda}(\cdot; t^{-1})$ is the operator of multiplication by the function defined in Remark $\ref{rem:other}$  (see also Proposition/Definition $\ref{df:Halls}$).
\end{thm}

\begin{rem}
From $(\ref{eqn:Dop})$ and \cite[(1.16)]{NS3}, the top-degree (of degree zero) of the interpolation Macdonald operator $D_N^k$ is the Macdonald $q$-difference operator $H_N^k$.
It follows that the degree zero part of $A^k$ diagonalizes the Macdonald functions and is given by
\begin{equation*}
A^k_{\textrm{Macdonald}} = \sum_{\ell(\lambda) = k}{t^{\lambda_1 + \lambda_2 + \dots}\PP_{\lambda}(\cdot; t^{-1})(\QQ_{\lambda}(\cdot; t^{-1}))^*},
\end{equation*}
because top homogeneous part of $\F_{\lambda}(\cdot; t^{-1})$ is $\PP_{\lambda}(\cdot; t^{-1})$.
The Macdonald operators at infinity, in \cite{NS3}, are expressed slightly differently.
Let us show how to obtain their formula from $(\ref{eqn:expansion})$, $(\ref{eqn:kexplicit})$.
Since the Macdonald function $M_{\lambda}(\cdot; q, t)$ is invariant under the simultaneous change of parameters $(q, t) \leftrightarrow (1/q, 1/t)$, the top-degree of the operator $A_{\infty}(u; 1/q, 1/t)$ also diagonalizes the Macdonald functions.
These are, in fact, the operators considered in \cite{NS3}.
From Theorem $\ref{thm:main}$, we can write the top-degree part of $A_{\infty}(u; 1/q, 1/t)$ as
\begin{equation*}
1 + \frac{\widetilde{A}_{\textrm{Macdonald}}^1}{(u; t^{-1})_1} + \frac{\widetilde{A}_{\textrm{Macdonald}}^2}{(u; t^{-1})_2} + \dots,
\end{equation*}
where
\begin{equation}\label{eqn:almostNS}
\widetilde{A}^k_{\textrm{Macdonald}} = \sum_{\ell(\lambda) = k}{t^{-\lambda_1 - \lambda_2 - \dots}\PP_{\lambda}(\cdot; t)(\QQ_{\lambda}(\cdot; t))^{-*}},
\end{equation}
and $f^{-*}$ is the adjoint of multiplication by $f\in\Lambda_{\Q(q, t)}$ with respect to the inner product $(\cdot, \cdot)_{1/q, 1/t}$ determined by
\begin{equation*}
(p_{\lambda}, p_{\mu})_{1/q, 1/t} = \delta_{\mu, \nu}\prod_{i\geq 1}{(i^{m_i(\mu)}m_i(\mu)!)}\cdot\prod_{i=1}^{\ell(\mu)}{\frac{1 - (1/q)^{\mu_i}}{1 - (1/t)^{\mu_i}}} = (tq^{-1})^{\mu_1 + \mu_2 + \dots}(p_{\lambda}, p_{\mu})_{q, t} \ \forall \lambda, \mu \in \Y.
\end{equation*}
It follows that $p_n^{-*} = (t/q)^n p_n^*$ for all $n\geq 1$ and, more generally, $f^{-*} = (t/q)^{\deg f} f^*$ for all homogeneous $f\in\Lambda_{\Q(q, t)}$. It follows that $(\ref{eqn:almostNS})$ equals
\begin{equation*}
\widetilde{A}^k_{\textrm{Macdonald}} = \sum_{\ell(\lambda) = k}{q^{-\lambda_1 - \lambda_2 - \dots}\PP_{\lambda}(\cdot; t)(\QQ_{\lambda}(\cdot; t))^{*}},
\end{equation*}
which is exactly the formula in the Theorem of \cite{NS3}.
\end{rem}

\section{Proof of the Main Theorem}\label{sec:mainproof}

In this section, we prove Theorem $\ref{thm:main}$.
As in \cite{NS1, NS3}, the statement of Theorem $\ref{thm:main}$ reduces to a combinatorial identity between the families of symmetric polynomials/functions $\{\QQ_{\lambda}\}_{\lambda}$ and $\{\F_{\lambda}\}_{\lambda}$.
The identity to prove is the \textit{refined Cauchy identity} in Proposition $\ref{prop:refined}$.
Our proof of this proposition is new and yields, as a corollary, the main results of the aforementioned papers.

\subsection{Formalities on completed tensor products}\label{sec:tensors}

We have been using the sequence of variables $X = (x_1, x_2, \ldots)$ and $\Lambda_{\FF}$ for the algebra of symmetric functions on the set of variables $X$.
We shall need a second sequence of variables $Y = (y_1, y_2, \ldots)$, which is why we write $\Lambda_{X, \FF}$ and $\Lambda_{Y, \FF}$ to distinguish the corresponding algebras of symmetric functions.
We also write $\Lambda_{X, N, \FF}$ for the algebra of polynomials on $x_1, \ldots, x_N$.

As it is usual, $\Lambda_{X, \FF}\otimes\Lambda_{Y, \FF}$ is the tensor product of these algebras, whose elements are of the form
\begin{equation}\label{eqn:tensor}
\sum_{\lambda\in\Y}{c_{\lambda}(a_{\lambda}(x_1, x_2, \ldots)\otimes b_{\lambda}(y_1, y_2, \ldots))}, \ c_{\lambda}\in\FF,
\end{equation}
$\{a_{\lambda}(x_1, x_2, \ldots) : \lambda\in\Y\}$, $\{b_{\lambda}(y_1, y_2, \ldots) : \lambda\in\Y\}$ are bases of $\Lambda_{X, \FF}$ and $\Lambda_{Y, \FF}$, and $c_{\lambda} = 0$ for all but finitely many $\lambda\in\Y$.

Moreover we need the \textit{completed tensor product} $\Lambda_{X, \FF} \widehat{\otimes} \Lambda_{Y, \FF}$, which is the algebra whose elements are of the form $(\ref{eqn:tensor})$, except that now $c_{\lambda}$ can be nonzero for infinitely many $\lambda\in\Y$.
The most important element of $\Lambda_{X, \FF} \widehat{\otimes} \Lambda_{Y, \FF}$, for our purposes, is the \textit{(Macdonald) reproducing kernel}
\begin{equation*}
\Pi := \prod_{i \geq 1}\prod_{j \geq 1} {\frac{(tx_iy_j; q)_{\infty}}{(x_iy_j; q)_{\infty}}},
\end{equation*}
where $(z; q)_{\infty} := (1 - z)(1 - zq)\cdots$ is the infinite Pochhammer symbol.
For example, by the \textit{Cauchy identity for Macdonald polynomials}, \cite[(4.13)]{M}, we can write it as
\begin{equation*}
\Pi = \sum_{\lambda\in\Y} {P_{\lambda}(x_1, x_2, \ldots; q, t)Q_{\lambda}(y_1, y_2, \ldots; q, t)}.
\end{equation*}

Given an operator $A$ on $\Lambda_{\FF} = \Lambda_{X, \FF}$, we can easily extend it to the tensor product $\Lambda_{X, \FF}\otimes\Lambda_{Y, \FF}$ and to the completed tensor product $\Lambda_{X, \FF}\widehat{\otimes}\Lambda_{Y, \FF}$, by making the operator act on the first coordinate.
Similarly we can extend an operator $B$ on $\Lambda_{\FF} = \Lambda_{Y, \FF}$ to $\Lambda_{X, \FF}\otimes\Lambda_{Y, \FF}$ and $\Lambda_{X, \FF}\widehat{\otimes}\Lambda_{Y, \FF}$ by making the operator act on the second coordinate.
The new operators are denoted by the same letter $A$ or $B$. This construction is used many times below.

\subsection{Reduction to an identity of symmetric functions}

\begin{lem}[Lemma from \cite{NS3}]\label{lem:NS3}
Let $f = f(x_1, x_2, \ldots) \in \Lambda_{\Q(q, t)}$ be arbitrary. Also let $f$ denote the operator $\Lambda_{\Q(q, t)} \rightarrow \Lambda_{\Q(q, t)}$ of multiplication by $f$, and let $f^*$ be the adjoint of $f$ with respect to the Macdonald inner product $(\cdot, \cdot)_{q, t}$ in $(\ref{eqn:macdonaldinner})$. Then
\begin{equation*}
f^* \left( \Pi \right) = f(y_1, y_2, \ldots)\cdot\Pi.
\end{equation*}
In the above equation, the operator $f^*$ and the operator of multiplication by $f(y_1, y_2, \ldots)$ act on $\Lambda_{X, \Q(q, t)}\widehat{\otimes}\Lambda_{Y, \Q(q, t)}$, as explained at the end of subsection $\ref{sec:tensors}$.
\end{lem}

Assume that we have an operator $A$ on $\Lambda_{\Q(q, t)} = \Lambda_{X, \Q(q, t)}$ such that $A(\Pi) = 0$. This implies
\begin{equation*}
A \left\{ \sum_{\lambda\in\Y} {P_{\lambda}(x_1, x_2, \ldots; q, t)Q_{\lambda}(y_1, y_2, \ldots; q, t)} \right\} = \sum_{\lambda\in\Y} {A(P_{\lambda}(x_1, x_2, \ldots; q, t))Q_{\lambda}(y_1, y_2, \ldots; q, t)} = 0,
\end{equation*}
and so $AP_{\lambda}(\cdot; q, t) = 0$, for all $\lambda\in\Y$. Since $\{P_{\lambda}(\cdot; q, t) : \lambda\in\Y\}$ is a basis of $\Lambda_{\Q(q, t)}$, then $A = 0$.

From the previous discussion and Lemma $\ref{lem:NS3}$, the main theorem is reduced to
\begin{equation}\label{eqn:toprove1}
A^k(\Pi) = \sum_{\ell(\lambda) = k}{ t^{\lambda_1 + \lambda_2 + \ldots} \F_{\lambda}(x_1, x_2, \ldots; t^{-1})\QQ_{\lambda}(y_1, y_2, \ldots; t^{-1}) }\cdot\Pi,
\end{equation}
as an equality in $\Lambda_{X, \Q(q, t)}\widehat{\otimes}\Lambda_{Y, \Q(q, t)}$.
Next, multiply both sides of $(\ref{eqn:toprove1})$ by $(u; t)_k^{-1} = ((1 - u)(1 - tu)\cdots (1 - t^{k-1}u))^{-1}$ and add from $k = 0$ to $\infty$.
It is then clear that $(\ref{eqn:toprove1})$ holds if and only if
\begin{equation}\label{eqn:toprove15}
A_{\infty}(u; q, t)(\Pi) = \sum_{k=0}^{\infty}{ \frac{1}{(u; t)_k} \sum_{\ell(\lambda) = k}{ t^{\lambda_1 + \lambda_2 + \ldots} \F_{\lambda}(x_1, x_2, \ldots; t^{-1}) \QQ_{\lambda}(y_1, y_2, \ldots; t^{-1}) } }\cdot \Pi
\end{equation}
holds as an equality in $\Lambda_{X, \Q(q, t)}\widehat{\otimes}\Lambda_{Y, \Q(q, t)}[[u]]$ (after Taylor expanding each $1/(u; t)_k$ near $u = 0$).
Our goal is to prove $(\ref{eqn:toprove15})$. In the next subsection, we make a further reduction of this equality to a combinatorial identity involving only HL functions and the inhomogeneous HL polynomials.

\subsection{Reduction to a refined Cauchy identity}

Let us begin with an observation. If $G\in\Lambda_{\Q(t)}$ is such that $\pi^{\infty}_N G = 0$ holds for all $N\in\N$, then $G = 0$ (recall the maps $\pi^{\infty}_N : \Lambda_{\Q(t)} \rightarrow \Lambda_{N, \Q(t)}$ were defined in $(\ref{eqn:mapsinfty})$).
Indeed, this can be deduced from the proof of Proposition/Definition $\ref{df:Halls}$.
As discussed at the end of subsection $\ref{sec:tensors}$, each operator $\pi^{\infty}_N$ can be extended to an operator of the form $\Lambda_{X, \Q(t)}\widehat{\otimes}\Lambda_{Y, \Q(t)}[[u]] \rightarrow \Lambda_{X, N, \Q(t)}\widehat{\otimes}\Lambda_{Y, \Q(t)}[[u]]$ by acting on the first coordinate (of each coefficient of a power $u^k$).
A similar statement holds: if $G\in \Lambda_{X, \Q(t)}\widehat{\otimes}\Lambda_{Y, \Q(t)}[[u]]$ is such that $\pi^{\infty}_N G = 0$ for all $N\in\N$, then $G = 0$.

Therefore it follows that $(\ref{eqn:toprove15})$ holds if, for each $N\in\N$, it also holds after applying to it the operator $\pi^{\infty}_N$.
As for the left side, we have $\pi^{\infty}_NA_{\infty}(u; q, t)(\Pi) = A_N(u; q, t)\pi^{\infty}_N(\Pi)$, by Lemma $\ref{lem:opinfty}$.
Next, using the well-known identity \cite[Ch. VI, (2.6)]{M}, we have
\begin{equation}\label{gather:eqns}
\begin{gathered}
\pi^{\infty}_N(\Pi) = \pi^{\infty}_N\left( \exp\left( \sum_{n=1}^{\infty}{ \frac{1 - t^n}{n(1 - q^n)} p_n(x_1, x_2, \ldots)p_n(y_1, y_2, \ldots) } \right) \right)\\
= \exp\left( \sum_{n=1}^{\infty}{ \frac{1 - t^n}{n(1 - q^n)} (p_n(x_1, \ldots, x_N) + \frac{t^{-nN}}{1-t^{-n}})p_n(y_1, y_2, \ldots) } \right) = \\
\exp\left( \sum_{n=1}^{\infty}{ \frac{1 - t^n}{n(1 - q^n)} p_n(x_1, \ldots, x_N)p_n(y_1, y_2, \ldots) } \right)
\exp\left( -\sum_{n=1}^{\infty}{ \frac{p_n(t^{1-N}y_1, t^{1-N}y_2, \ldots)}{n(1 - q^n)} } \right)\\
= \Pi_N \cdot \exp\left( -\sum_{n=1}^{\infty}{ \frac{p_n(t^{1-N}y_1, t^{1-N}y_2, \ldots)}{n(1 - q^n)} } \right) = \Pi_N \cdot\prod_{i=1}^N{\frac{1}{(t^{1-N}y_i; q)_{\infty}}},
\end{gathered}
\end{equation}
where $\Pi_N$ denotes
\begin{equation*}
\Pi_N := \prod_{i=1}^N\prod_{j\geq 1}{\frac{(tx_iy_j; q)_{\infty}}{(x_iy_j; q)_{\infty}}}.
\end{equation*}
We note that the second to last equality in the display $(\ref{gather:eqns})$ follows from \cite[Ch. VI, (2.6)]{M} after setting $x_{N+1} = x_{N+2} = \ldots = 0$, whereas the last equality in display $(\ref{gather:eqns})$ follows from the same identity but now after setting $t = 0$, and then using the variables $t^{1-N}y_i$ instead of $y_i$.

Next we find an expression for the right side of $(\ref{gather:eqns})$ after applying $\pi^{\infty}_N$ to it.
Since $\pi^{\infty}_N$ is a homomorphism of algebras, we obtain the same right side, except that we replace $\F_{\lambda}(x_1, x_2, \ldots; t^{-1})$ and $\Pi$ by
\begin{equation*}
\pi^{\infty}_N\F_{\lambda}(x_1, x_2, \ldots; t^{-1}) = \F_{\lambda}(x_1, \ldots, x_N; t^{-1}), \hspace{.2in} \pi^{\infty}_N(\Pi) = \Pi_N\cdot\prod_{i=1}^N{\frac{1}{(t^{1-N}y_i; q)_{\infty}}}.
\end{equation*}
After factoring out the factor $\prod_{i=1}^N{(t^{1-N}y_i; q)_{\infty}^{-1}}$ from both sides, we have that the result of applying $\pi^{\infty}_N$ to $(\ref{eqn:toprove2})$ is equivalent to
\begin{equation}\label{eqn:almostrefined}
\Pi_N^{-1}A_N(u; q, t)(\Pi_N) = \sum_{k=0}^{\infty}{\frac{1}{(u; t)_k} \sum_{\ell(\lambda) = k}{ t^{\lambda_1 + \lambda_2 + \ldots} \F_{\lambda}(x_1, \ldots, x_N; t^{-1})\QQ_{\lambda}(y_1, y_2, \ldots; t^{-1}) } }.
\end{equation}

We can still simplify the left side of $(\ref{eqn:almostrefined})$. After expanding the determinant in $(\ref{eqn:Aop})$, we have
\begin{equation*}
A_N(u; q, t) = \frac{1}{V(x_1, \ldots, x_N)(u; t)_N}\sum_{w\in S_N}
\epsilon(w) \prod_{i=1}^N{ x_i^{N - w(i) - 1} \{(x_i - t^{N-1})T_{q, x_i} + t^{w(i) - 1}(1 - x_i u)\} }
\end{equation*}
where $\epsilon(w)$ is the signature of the permutation $w$. Also note that all $N$ operators (indexed by $i = 1, \ldots, N$) pairwise commute, so the order in the product does not matter. From the evident
\begin{equation*}
\Pi_N^{-1}T_{q, x_i}\Pi_N = \prod_{l = 1}^{\infty}{ \frac{1 - x_i y_l}{1 - t x_i y_l} },
\end{equation*}
we deduce
\begin{gather}
\Pi_N^{-1} A_N(u; q, t) (\Pi_N) = \frac{1}{V(x_1, \ldots, x_N)(u; t)_N}\nonumber\\
\times \sum_{w\in S_N}
\epsilon(w) \prod_{i=1}^N{ x_i^{N - w(i) - 1} \{(x_i - t^{N-1})\prod_{l = 1}^{\infty}{ \frac{1 - x_i y_l}{1 - t x_i y_l} } + t^{w(i) - 1}(1 - x_i u)\} }\nonumber\\
= \frac{1}{V(x_1, \ldots, x_N)(u; t)_N}\det_{1\leq i, j\leq N}\left[ x_j^{N-i-1}\{ (x_j - t^{N-1})\prod_{l = 1}^{\infty}{\frac{1 - x_jy_l}{1 - tx_jy_l}} + t^{i-1}(1 - x_ju) \} \right]\label{eqn:leftside}
\end{gather}

To summarize, Theorem $\ref{thm:main}$ has been reduced to prove that $(\ref{eqn:leftside})$ equals the right side of $(\ref{eqn:almostrefined})$, for any $N\in\N$.
This equality follows from Proposition $\ref{prop:refined}$ below, after sending $y_i$ to $t^{-1}y_i$, using the homogeneity $t^{\lambda_1 + \lambda_2 + \ldots}\QQ_{\lambda}(t^{-1}y_1, t^{-1}y_2, \ldots; t^{-1}) = \QQ_{\lambda}(y_1, y_2, \ldots; t^{-1})$, and then replacing $t$ by $t^{-1}$.
Equalities of this sort were called \textit{refined Cauchy identities} in \cite{WZJ}.
Some discussion of their work, in connection to the papers \cite{NS3, W} and ours, is given in subsection $\ref{sec:discussion}$.

\subsection{Proof of the refined Cauchy identity}

The goal of this subsection is to prove the following proposition which, by the arguments in the previous subsection, concludes the proof of the main theorem. Observe that the parameter $q$ that gave rise to the Macdonald inner product and to the adjoint in formula $(\ref{eqn:kexplicit})$ is gone.

\begin{prop}\label{prop:refined}
As an identity on $\Lambda_{X, N, \Q(t)}\widehat{\otimes}\Lambda_{Y, \Q(t)}[[u]]$, the following holds
\begin{equation}\label{eqn:refined}
\begin{gathered}
1 + \sum_{k=1}^N{\frac{1}{(u; t^{-1})_k}\sum_{\ell(\lambda) = k}{\F_{\lambda}(x_1, \dots, x_N; t)\QQ_{\lambda}(y_1, y_2, \dots; t)}} =\\
\frac{1}{V(x_1, \dots, x_N)\cdot(u; t^{-1})_N}\det_{1\leq i, j\leq N}\left[ x_j^{N - i - 1}\left\{ (x_j - t^{1 - N})\prod_{l = 1}^{\infty}{\frac{1 - tx_jy_l}{1 - x_jy_l}} + t^{1 - i}(1 - x_ju) \right\} \right].
\end{gathered}
\end{equation}
\end{prop}

Denote by $M$ the $N\times N$ matrix on the right side of $(\ref{eqn:refined})$. For any subset $I \subset \{1, 2, \ldots, N\}$, define the $N\times N$ matrix $M_I$ by
\begin{equation}
 (M_I)_{i, j} :=
        \begin{cases}
            x_j^{N - i - 1}(x_j - t^{1 - N})\prod_{l = 1}^{\infty}{\frac{1 - tx_jy_l}{1 - x_jy_j}}, & \text{ if }j\in I,\\
            x_j^{N-i-1}t^{1-i}(1 - x_ju), & \text{ if }j \notin I.
        \end{cases}
\end{equation}

By factoring out common factors from each column, the determinant of $M_I$ equals
\begin{equation}\label{eqn:detMI}
\det(M_I) = \prod_{l = 1}^{\infty}\prod_{i\in I}\frac{1 - tx_iy_l}{1 - x_iy_l}
\times \prod_{i\in I} x_i^{-1}(x_i - t^{1 - N})
\times \prod_{j\notin I} (x_j^{-1}t^{1-N}(1 - x_j u))
\times \det(A_I),
\end{equation}
where the $N\times N$ matrix $A_I$ is
\begin{equation*}
(A_I)_{i, j} := 
	\begin{cases}
            x_j^{N - i}, & \text{ if }j\in I,\\
            (tx_j)^{N-i}, & \text{ if }j \notin I.
        \end{cases}
\end{equation*}
Observe that $A_I$ is the Vandermonde determinant on variables $\{x_j : j \in I\} \sqcup \{tx_j : j\notin I\}$,
and the ordering of these variables (in the matrix) is inherited from $1 < 2 < \dots < N$.
We deduce
\begin{equation}\label{eqn:detAI}
\frac{\det(A_I)}{V(x_1, \ldots, x_N)} = t^{{N - |I| \choose 2}}\prod_{i\in I}\prod_{j\notin I}{ \frac{x_i - tx_j}{x_i - x_j} }.
\end{equation}

Plugging $(\ref{eqn:detAI})$ and $(\ref{eqn:cauchyHL})$ into $(\ref{eqn:detMI})$, we obtain
\begin{equation}\label{eqn:detVander}
\begin{gathered}
\frac{\det(M_I)}{V(x_1, \ldots, x_N)} = t^{{N - |I| \choose 2}} \prod_{i\in I}{(1 - t^{1-N}x_i^{-1})}
\prod_{j\notin I}{(t^{1-N}(x_j^{-1} - u))}
\prod_{\substack{i \in I \\ j \notin I}}{ \frac{x_i - tx_j}{x_i - x_j} }\\
\times\sum_{\ell(\lambda) \leq |I|}{ \PP_{\lambda}(\{x_i : i\in I\}; t) \QQ_{\lambda}(y_1, y_2, \ldots; t) }.
\end{gathered}
\end{equation}
From the definition of the matrix $M_I$ and multilinearity of the determinant, we have $\det(M) = \sum_{I\subseteq\{1, \ldots, N\}}{\det(M_I)}$. Therefore the expression in the right side of $(\ref{eqn:refined})$ equals
\begin{equation}\label{eqn:RHSrefined}
\begin{gathered}
\frac{1}{(u; t^{-1})_N}\sum_{I\subseteq\{1, \ldots, N\}}
t^{{N - |I| \choose 2}}
\prod_{i\in I}{(1 - t^{1-N}x_i^{-1})}
\prod_{j\notin I}{(t^{1-N}(x_j^{-1} - u))}
\prod_{\substack{i\in I \\ j\notin I}}{ \frac{x_i - tx_j}{x_i - x_j} }\\
\sum_{\ell(\lambda) \leq |I|}{ \PP_{\lambda}(\{x_i : i\in I\}; t) \QQ_{\lambda}(y_1, y_2, \ldots; t) }.
\end{gathered}
\end{equation}

We must prove that  $(\ref{eqn:RHSrefined})$ equals the left side of $(\ref{eqn:refined})$.
Both expressions are of the form $\sum_{\lambda\in\Y(N)}{G_{\lambda}(x_1, \ldots, x_N; t)\QQ_{\lambda}(y_1, y_2, \ldots; t)}$.
Thus let us choose any $\lambda\in\Y(N)$ of length $0 \leq \ell(\lambda) = k\leq N$ and show that the symmetric polynomial on $x_1, \ldots, x_N$ that accompanies $\QQ_{\lambda}(y_1, y_2, \ldots; t)$ in $(\ref{eqn:RHSrefined})$ equals $\F_{\lambda}(x_1, \ldots, x_N; t)/(u; t^{-1})_k$, which accompanies $\QQ_{\lambda}(y_1, y_2, \ldots; t)$ in the left side of $(\ref{prop:refined})$.
After simple algebraic manipulations, the identity to prove becomes
\begin{equation}\label{eqn:toprove}
\begin{gathered}
\sum_{\substack{ I \subseteq \{1, \ldots, N\} \\  |I| \geq k}}
t^{-\frac{(N + |I| - 1)(N - |I|)}{2}}
\prod_{i\in I}{(1 - t^{1-N}x_i^{-1})}
\prod_{j\notin I}{(x_j^{-1} - u)}
\prod_{\substack{ i\in I \\ j \notin I }}{ \frac{x_i - tx_j}{x_i - x_j} }
\PP_{\lambda}(\{x_i : i\in I\}; t)\\
\stackrel{?}{=} \frac{(u; t^{-1})_N}{(u; t^{-1})_k}\F_{\lambda}(x_1, \ldots, x_N; t).
\end{gathered}
\end{equation}

Before proving $(\ref{eqn:toprove})$ in full generality, let us look first at the extreme cases.\\

\textbf{Case 1.} $\ell(\lambda) = k = N$. In this case, the sum in the left side of $(\ref{eqn:toprove})$ has only one term corresponding to $I = \{1, 2, \ldots, N\}$. Such term is
\begin{equation*}
\prod_{i=1}^N{(1 - t^{1-N}x_i^{-1})}\PP_{\lambda}(x_1, \ldots, x_N; t).
\end{equation*}
The latter expression equals $\F_{\lambda}(x_1, \ldots, x_N; t)$, see Remark $\ref{rem:specialcases}$. Therefore $(\ref{eqn:toprove})$ holds in this case.\\

\textbf{Case 2.} $\ell(\lambda) = 0$, i.e., $\lambda = \emptyset$. In this case we can use $\PP_{\emptyset}(\{x_i : i\in I\}; t) = \F_{\lambda}(x_1, \ldots, x_N; t) = 1$, so that identity $(\ref{eqn:toprove})$ is easily deduced from the following lemma for $K = N$, $A = \{1, 2, \ldots, N\}$, and under the identification of variables $x_i \leftrightarrow t^{1-N}z_i^{-1}\ \forall i = 1, \ldots, N$, $u \leftrightarrow t^{N-1}q$. The proof of Lemma $\ref{lem:technical}$ is postponed till the end of this subsection.

\begin{lem}\label{lem:technical}
Let $K\in\N$, let $A$ be a set of size $K$, and let $(z_a)_{a\in A}$ be a set of variables indexed by $A$.
Moreover $q, t$ are two additional formal parameters.
Then
\begin{equation}\label{eqn:lemma}
\sum_{J\subseteq A}{ t^{{K - |J| \choose 2}}\prod_{i\in J}{(1 - z_i)}\cdot \prod_{j\in A\setminus J}{(z_j - q)} \cdot \prod_{\substack{i\in J \\ j\in A\setminus J}}{\frac{tz_i - z_j}{z_i - z_j}} } = (q; t)_K.
\end{equation}
\end{lem}

\textbf{General case.} $1 \leq \ell(\lambda) = k \leq N-1$. The idea is to use the definition of the HL polynomial for $\PP_{\lambda}(\{x_i : i\in I\})$ and the definition of the inhomogeneous HL polynomial for $\F_{\lambda}(x_1, \ldots, x_N; t)$, then write both sides of $(\ref{eqn:toprove})$ as big sums and match terms of these sums with the help of Lemma $\ref{lem:technical}$.

The prefactors for both $\PP_{\lambda}(\{x_i : i\in I\}; t)$ and $\F_{\lambda}(x_1, \ldots, x_N; t)$ are very similar and almost match each other, except for factor corresponding to the part $0$ of the partition.
For $\PP_{\lambda}(\{x_i : i\in I\}; t)$, this factor is $\prod_{j=1}^{|I| - k}{(1-t)/(1-t^j)}$, whereas for $\F_{\lambda}(x_1, \ldots, x_N; t)$, this factor is $\prod_{j=1}^{N-k}{(1-t)/(1-t^j)}$. For any $I\subset\{1, \ldots, N\}$, let $S_I$ be the group of permutations of elements of $I$; it is finite of size $|I|!$.
Also denote by $i_1 < i_2 < \ldots < i_k$ the smallest elements of $I$.
With these considerations, the equality $(\ref{eqn:toprove})$ we wish to prove becomes
\begin{equation}\label{eqn:toprove2}
\begin{gathered}
\sum_{\substack{ I \subseteq \{1, \ldots, N\} \\  |I| \geq k}}
\sum_{w\in S_I}
t^{-\frac{(N + |I| - 1)(N - |I|)}{2}}
\prod_{j=1}^{|I|-k}{\frac{1 - t}{1 - t^j}}
\prod_{i\in I}{(1 - t^{1-N}x_i^{-1})}
\prod_{j\notin I}{(x_j^{-1} - u)}\\
\times\prod_{\substack{ i\in I \\ j \notin I }}{ \frac{x_i - tx_j}{x_i - x_j} }
\ w\left\{x_{i_1}^{\lambda_1}\cdots x_{i_k}^{\lambda_k}\prod_{\substack{i, j\in I \\ i < j}}{ \frac{x_i - tx_j}{x_i - x_j} } \right\}
\stackrel{?}{=} \frac{(u; t^{-1})_N}{(u; t^{-1})_k}
\prod_{j=1}^{N-k}{\frac{1 - t}{1 - t^j}}\\
\times\sum_{\sigma\in S_N}
\sigma\left\{ \prod_{i=1}^k{(1 - t^{1-N}x_i^{-1})}x_1^{\lambda_1}\cdots x_k^{\lambda_k}
\prod_{1\leq i<j\leq N}{\frac{x_i - tx_j}{x_i - x_j}} \right\}.
\end{gathered}
\end{equation}

Next we break each of the two sides of the equality above into $N(N-1)\cdots (N-k+1)$ terms, and match those on the left with those on the right. Let
\begin{equation*}
R := (r_1, \ldots, r_k) \in \{1, 2, \ldots, N\}^{k}
\end{equation*}
be an arbitrary tuple of size $k$. For our argument, let us fix one such $k$-tuple $R$.
We say that $\sigma\in S_N$ is $R$-restricted if $\sigma(1) = r_1, \ldots, \sigma(k) = r_k$.
The right side of $(\ref{eqn:toprove2})$ with $\sigma\in S_N$ replaced by only those $R$-restricted $\sigma\in S_N$ is called the $R$-restricted right side of $(\ref{eqn:toprove2})$. Next consider a pair $(I \subseteq \{1, \ldots, N\}, \ w\in S_N)$ such that $|I| \geq k$ and $I = \{i_1 < \ldots < i_k < \ldots\}$ are the elements of $I$ in increasing order.
We say that $(I, w)$ is $R$-restricted if $\{r_1, \ldots, r_k\}\subset I$ and $w(i_1) = r_1, \ldots, w(i_k) = r_k$.
The left side of $(\ref{eqn:toprove2})$ with the double sum over $(I, w)$ replaced by only those $R$-restricted pairs is called the $R$-restricted left side of $(\ref{eqn:toprove2})$.

Let us simplify the $R$-restricted sides of the equality and then show they are equal to each other.
Begin with the right side.
Let $I_0 := \{r_1, \ldots, r_k\}$ and denote by $S(I_0)$ be the set of bijective mappings $\sigma' : \{k+1, \ldots, N\} \rightarrow \{1, \ldots, N\} \setminus \{r_1, \ldots, r_k\}$.
Then the $R$-restricted right side of $(\ref{eqn:toprove2})$, corresponding to $R = (r_1, \ldots, r_k)$, equals
\begin{equation}\label{rhs:desired}
\begin{gathered}
\frac{(u; t^{-1})_N}{(u; t^{-1})_k}\prod_{i\in I_0}{(1 - t^{1-N}x_i^{-1})}
\prod_{i=1}^k{x_{r_i}^{\lambda_i}}
\prod_{1\leq i<j\leq k}{\frac{x_{r_i} - tx_{r_j}}{x_{r_i} - x_{r_j}}}
\prod_{\substack{i\in I_0 \\ j\notin I_0}}{\frac{x_i - tx_j}{x_i - x_j}}\\
\times\prod_{j=1}^{N-k}{\frac{1 - t}{1 - t^j}}
\sum_{\sigma'\in S(I_0)}{\sigma'\left\{ \prod_{k+1\leq i<j\leq N}{\frac{x_i - tx_j}{x_i - x_j}} \right\}}.
\end{gathered}
\end{equation}
The second line in the display $(\ref{rhs:desired})$ is equal to $1$ by virtue of $(\ref{eqn:combinatorial})$.
It follows that the $R$-restricted right side of $(\ref{eqn:toprove2})$ is equal to the first line in the display $(\ref{rhs:desired})$.

We switch to simplifying the $R$-restricted left side of $(\ref{eqn:toprove2})$. Again set $I_0 := \{r_1, \ldots, r_k\}$. It is clear that
\begin{equation*}
\begin{aligned}
w\left\{x_{i_1}^{\lambda_1}\cdots x_{i_k}^{\lambda_k}\prod_{\substack{i, j\in I \\ i < j}}{ \frac{x_i - tx_j}{x_i - x_j} } \right\}
=& \prod_{i=1}^k{x_{r_i}^{\lambda_i}}\prod_{1\leq i<j\leq k}{\frac{x_{r_i} - tx_{r_j}}{x_{r_i} - x_{r_j}}}\prod_{\substack{i\in I_0 \\ j\in I\setminus I_0}}{\frac{x_i - tx_j}{x_i - x_j}}\\
&\times w\left\{ \prod_{\substack{i, j\in I\setminus\{i_1, \ldots, i_k\} \\ i < j}}{\frac{x_i - tx_j}{x_i - x_j}} \right\}.
\end{aligned}
\end{equation*}
Let $S(\{i_1, \ldots, i_k\}, I_0)$ be the set of bijective maps $I\setminus\{i_1, \ldots, i_k\} \rightarrow I \setminus I_0$; the restriction of $w\in S_N$ to $I\setminus\{i_1, \ldots, i_k\}$, to be denoted $w'$, is an element of $S(\{i_1, \ldots, i_k\}, I_0)$. Now we can write the $R$-restricted left side of $(\ref{eqn:toprove2})$ as
\begin{equation}\label{eqn:largedisplay}
\begin{gathered}
\prod_{i=1}^k{x_{r_i}^{\lambda_i}}\prod_{1\leq i<j\leq k}{\frac{x_{r_i} - tx_{r_j}}{x_{r_i} - x_{r_j}}}
\sum_{I_0 \subseteq I \subseteq \{1, \ldots, N\}}
t^{-\frac{(N + |I| - 1)(N - |I|)}{2}}
\prod_{\substack{i\in I_0 \\ j\in I\setminus I_0}}{\frac{x_i - tx_j}{x_i - x_j}}
\prod_{i\in I}{(1 - t^{1-N}x_i^{-1})}\\
\prod_{j\notin I}{(x_j^{-1} - u)}
\prod_{\substack{ i\in I \\ j \notin I }}{ \frac{x_i - tx_j}{x_i - x_j} }
\times\prod_{j=1}^{|I|-k}{\frac{1 - t}{1 - t^j}}\sum_{w' \in S(\{i_1, \ldots, i_k\}, I_0)}
w'\left\{ \prod_{\substack{i, j\in I\setminus\{i_1, \ldots, i_k\} \\ i < j}}{\frac{x_i - tx_j}{x_i - x_j}} \right\}\\
= \prod_{i=1}^k{x_{r_i}^{\lambda_i}}\prod_{1\leq i<j\leq k}{\frac{x_{r_i} - tx_{r_j}}{x_{r_i} - x_{r_j}}}\\
\times\sum_{I_0 \subseteq I \subseteq \{1, \ldots, N\}}
t^{-\frac{(N + |I| - 1)(N - |I|)}{2}}
\prod_{\substack{i\in I_0 \\ j\in I\setminus I_0}}{\frac{x_i - tx_j}{x_i - x_j}}
\prod_{i\in I}{(1 - t^{1-N}x_i^{-1})}
\prod_{j\notin I}{(x_j^{-1} - u)}
\prod_{\substack{ i\in I \\ j \notin I }}{ \frac{x_i - tx_j}{x_i - x_j} }\\
= \prod_{i=1}^k{x_{r_i}^{\lambda_i}}\prod_{1\leq i<j\leq k}{\frac{x_{r_i} - tx_{r_j}}{x_{r_i} - x_{r_j}}}\prod_{i\in I_0}{(1 - t^{1-N}x_i^{-1})}\prod_{\substack{i\in I_0 \\ j\notin I_0}}{\frac{x_i - tx_j}{x_i - x_j}}\times \sum_{I_1 \subseteq \{1, \ldots, N\}\setminus I_0} t^{{N - k - |I_0| \choose 2}}\\
\prod_{i\in I_1}{(1 - t^{1-N}x_i^{-1})}\prod_{j\in (\{1, \ldots, N\}\setminus I_0)\setminus I_1}{(t^{1-N}x_j^{-1} - ut^{1-N})}
\prod_{\substack{i\in I_1 \\ j\in (\{1, \ldots, N\}\setminus I_0)\setminus I_1}}{\frac{x_i - tx_j}{x_i - x_j}}\\
= \prod_{i=1}^k{x_{r_i}^{\lambda_i}}\prod_{1\leq i<j\leq k}{\frac{x_{r_i} - tx_{r_j}}{x_{r_i} - x_{r_j}}}\prod_{i\in I_0}{(1 - t^{1-N}x_i^{-1})}\prod_{\substack{i\in I_0 \\ j\notin I_0}}{\frac{x_i - tx_j}{x_i - x_j}}\times (ut^{1-N}; t)_{N - k},
\end{gathered}
\end{equation}
where the first equality in display $(\ref{eqn:largedisplay})$ follows from $(\ref{eqn:combinatorial})$, the second equality is a simple algebraic manipulation, and the third equality is a consequence of Lemma $\ref{lem:technical}$ applied to $K = N-k$, $A = \{1, \ldots, N\}\setminus I_0$, $z_i \leftrightarrow t^{1-N}x_i^{-1} \ \forall i = 1, \ldots, N$, and $q \leftrightarrow ut^{1-N}$.

Finally, observe $(ut^{1-N}; t)_{N - k} = (u; t^{-1})_N/(u; t^{-1})_k$, so the last line of display $(\ref{eqn:largedisplay})$ equals the first line of display $(\ref{rhs:desired})$. This implies that, for a fixed $R = (r_1, \ldots, r_k)$, the $R$-restricted left side and $R$-restricted right side of $(\ref{eqn:toprove2})$ are equal. Since this holds for any $R\in \{1, \ldots, N\}^{k}$, adding these equalities over all $N(N-1)\cdots (N-k+1)$ distinct $k$-tuples $R$, the identity $(\ref{eqn:toprove2})$ follows. It only remains to prove the key Lemma $\ref{lem:technical}$.

\begin{proof}[Proof of Lemma $\ref{lem:technical}$]
For convenience, let $A = \{1, 2, \ldots, K\}$, so the variables are $z_1, z_2, \ldots, z_K$.
Let us make the change of variables $z_i \leftrightarrow y_i^{-1}$ for $i = 1, 2, \ldots, K$ and let the sum run over subsets $I := \{1, \ldots, K\} \setminus J$. After minor algebraic manipulations, the identity to prove $(\ref{eqn:lemma})$ becomes
\begin{equation}\label{lemmaproof1}
\sum_{I\subseteq\{1, \ldots, K\}}
t^{{|I| \choose 2}}\prod_{i\notin I}{\frac{y_i - 1}{1 - y_i q}}
\prod_{\substack{i\in I \\ j \notin I}}{\frac{ty_i - y_j}{y_i - y_j}}
= (q; t)_K\cdot\prod_{i=1}^K{(y_i^{-1} - q)^{-1}}.
\end{equation}
We prove $(\ref{lemmaproof1})$ with the help of Macdonald $q$-difference operators $\ref{eqn:macdonaldops}$. Begin with the equality
\begin{equation*}
\prod_{i\in I}{T_{q, y_i}}\left\{\prod_{i=1}^K{(y_i - 1)}\right\} = \prod_{j\notin I}{(y_j - 1)}\prod_{i\in I}{(qy_i - 1)} = (-1)^{|I|}\prod_{i=1}^K{(1 - qy_i)}\prod_{i\notin I}{\frac{y_i - 1}{1 - y_iq}},
\end{equation*}
which implies that the left side of $(\ref{lemmaproof1})$ equals
\begin{equation*}
\begin{gathered}
\prod_{i=1}^K{(1 - qy_i)^{-1}}\times\left( \sum_{k=0}^K{(-1)^k \sum_{\substack{ I \subset \{1, \ldots, K\} \\ |I| = k}}
{t^{{k \choose 2}} \prod_{\substack{i\in I \\ j\notin I}}{\frac{ty_i - y_j}{y_i - y_j}}\prod_{i\in I}{T_{q, y_i}} }
} \right) \left( \prod_{i=1}^K{(y_i - 1)} \right)\\
= \prod_{i=1}^K{(1 - qy_i)^{-1}}\times\left( \sum_{k=0}^K{(-1)^k H_K^k} \right)\left( \prod_{i=1}^K{(y_i - 1)} \right)\\
= \prod_{i=1}^K{(1 - qy_i)^{-1}}\times\left( \sum_{k=0}^K{(-1)^k H_K^k} \right)\left( \sum_{i=0}^K{(-1)^{K-i}e_i(y_1, \ldots, y_K)} \right),
\end{gathered}
\end{equation*}
where $\{H_K^k : k = 1, \ldots, K\}$  are the Macdonald $q$-difference operators, $H_K^0 := 1$, and $\{e_i(y_1, \ldots, y_K) : i = 0, \ldots, K\}$ are the elementary symmetric polynomials. It follows that $(\ref{lemmaproof1})$ is equivalent to
\begin{equation}\label{lemmaproof2}
\left( \sum_{k=0}^K{(-1)^k H_K^k} \right)\left( \sum_{i=0}^K{(-1)^{K-i}e_i(y_1, \ldots, y_K)} \right) = (q; t)_K (y_1y_2\cdots y_N).
\end{equation}
The essential property of the Macdonald $q$-difference operators is that they diagonalize the Macdonald polynomials, in particular they diagonalize the elementary symmetric polynomials $e_i(y_1, \ldots, y_K) = P_{(1^i)}(y_1, \ldots, y_K)$, \cite[Ch. VI (4.8)]{M}.
The eigenvalue is also known; in fact, we have
\begin{equation*}
H_K^ke_i(y_1, \ldots, y_K) = e_k(qt^{K-1}, \ldots, qt^{K-i}, t^{K-i-1}, \ldots, 1)e_{i}(y_1, \ldots, y_K).
\end{equation*}
Thus the left side of $(\ref{lemmaproof2})$ equals
\begin{equation*}
\sum_{i=0}^K{e_i(y_1, \ldots, y_K)\sum_{k=0}^K{(-1)^ke_k(qt^{K-1}, \ldots, qt^{K-i}, t^{K-i-1}, \ldots, t, 1)}}
\end{equation*}
and because $\sum_{k=0}^K{(-1)^ke_k(a_1, \ldots, a_K)} = \prod_{i=1}^K{(1 - a_i)}$, the coefficient of $e_i(y_1, \ldots, y_K)$ vanishes if $0\leq i < K$ (because the factor $\prod_{s=1}^i{(1 - qt^{K-s})}\prod_{r=1}^{K-i}{(1 - q^{r-1})}$ contains $1 - q^0 = 0$ unless $i = K$). The coefficient of $e_K(y_1, \ldots, y_K) = (y_1y_2\cdots y_K)$ is $\prod_{s=1}^K{(1 - qt^{K-s})} = (q; t)_K$.
\end{proof}

\subsection{Some corollaries}\label{sec:discussion}

By equating the top homogeneous components of the equality in Proposition $\ref{prop:refined}$, and using that the top-degree homogeneous component of $\F_{\lambda}(x_1, \ldots, x_N; t)$ is $\PP_{\lambda}(x_1, \ldots, x_N; t)$, we obtain:

\begin{cor}\label{cor:refined}
The following is an equality in $\Lambda_{X, N, \Q(t)}\widehat{\otimes}\Lambda_{Y, \Q(t)}[[u]]$:
\begin{equation}\label{eqn:refinedcauchy}
\begin{gathered}
\frac{1}{V(x_1, \dots, x_N)(u; t^{-1})_N}\det_{1\leq i, j\leq N}\left[ x_j^{N - i}\left\{ \prod_{l = 1}^{\infty}{\frac{1 - t x_j y_l}{1 - x_j y_l}} - t^{1 - i} u \right\} \right]\\
= 1 + \sum_{k=1}^{N}{\frac{1}{(u; t^{-1})_k}\sum_{\ell(\lambda) = k}{ \PP_{\lambda}(x_1, \dots, x_N; t)\QQ_{\lambda}(y_1, y_2, \dots; t)}}.
\end{gathered}
\end{equation}
\end{cor}

An equivalent version of $(\ref{eqn:refinedcauchy})$ was needed for the result of Nazarov-Sklyanin, \cite{NS3}.
Their proof is different from ours; it uses induction on $N$, and the Pieri rule for HL polynomials.
We do not have a Pieri rule for the inhomogeneous HL polynomials, so we devised a different method.

The identity of Corollary $\ref{cor:refined}$ was also proved by Wheeler and Zinn-Justin \cite{WZJ} (who introduced the name \textit{refined Cauchy identity}) and by Warnaar \cite{W}. The identity proved in both of these papers is when $y_{N+1} = y_{N+2} = \dots$, but this turns out to be equivalent to $(\ref{eqn:refinedcauchy})$.
Comparing $(\ref{eqn:refinedcauchy})$ with the result of \cite{WZJ, W}, one obtains the following nontrivial equality of degree $N$ polynomials in $u$:

\begin{equation*}
\begin{gathered}
\frac{1}{V(x_1, \dots, x_N)}\det_{1\leq i, j\leq N}\left[ x_j^{N - i}\left\{ \prod_{l = 1}^N{\frac{1 - q x_j y_l}{1 - x_j y_l}} - q^{1 - i} u \right\} \right] \\
= 
\frac{\prod_{i = 1}^N\prod_{j=1}^N{(1 - qx_i y_j)}}{V(x_1, \dots, x_N)V(y_1, \dots, y_N)}\det_{1\leq i, j\leq N}\left[ \frac{1 - uq^{1 - N} + (uq^{1-N} - q)x_iy_j}{(1 - qx_iy_j)(1 - x_iy_j)} \right].
\end{gathered}
\end{equation*}

In a different direction, we can set $u = 0$ in Proposition $\ref{prop:refined}$. We obtain the following \textit{inhomogeneous Cauchy identity}.
\begin{cor}\label{cor:cauchy}
The following is an equality in $\Lambda_{X, N, \Q(t)}\widehat{\otimes}\Lambda_{Y, \Q(t)}$:
\begin{equation}\label{coreqn:cauchy}
\begin{gathered}
\sum_{\ell(\lambda) \leq N}{\F_{\lambda}(x_1, \dots, x_N; t)\QQ_{\lambda}(y_1, y_2, \dots; t)}\\
= \frac{1}{V(x_1, \dots, x_N)}\det_{1\leq i, j\leq N}\left[ x_j^{N - i - 1}\left\{ (x_j - t^{1 - N})\prod_{l = 1}^{\infty}{\frac{1 - tx_jy_l}{1 - x_jy_l}} + t^{1 - i} \right\} \right].
\end{gathered}
\end{equation}
\end{cor}

By equating the top-degree homogeneous components of both sides of identity $(\ref{coreqn:cauchy})$, we obtain the usual Cauchy identity $(\ref{eqn:cauchyHL})$. However, the right side of $(\ref{coreqn:cauchy})$ does not have the usual factorized form of a reproducing kernel.

\section{Inhomogeneous Hall-Littlewood polynomials}\label{sec:combinatorialprops}

In this section, we study the inhomogeneous HL polynomials $\F_{\lambda}(x_1, \ldots, x_N; t)$.
The main result is Theorem $\ref{cor:equalityHLs}$, which proves that they are limits of interpolation Macdonald polynomials in the regime $q\rightarrow 0$.
The statement of this corollary was conjectured by Grigori Olshanski, who also suggested a proof; the elaboration of this idea is in the first two subsections below.

\subsection{Expansion in the basis of Hall-Littlewood polynomials}\label{sec:Hall1}

For $n\in\N_0$, let
\begin{equation*}
\phi_n(t) := \begin{cases}
(1 - t)(1 - t^2)\cdots (1 - t^n), &\textrm{ if }n \geq 1,\\
1, &\textrm{ if }n = 0.
\end{cases}
\end{equation*}
The $t$-analogue of the factorial is $[n]! := \phi_n(t)/(1-t)^n$.
For integers $n \geq k \geq 0$, the $t$-binomial coefficient ${n \brack k}$ is
\begin{equation*}
{n \brack k} := \frac{[n]!}{[k]![n-k]!} = \frac{\phi_n(t)}{\phi_k(t)\phi_{n-k}(t)}.
\end{equation*}
For convenience, we extend the definition to all integers $k\in\Z$ by setting
\begin{equation*}
{n \brack k} := 0 \ \forall k\in\Z \textrm{ with } k < 0 \textrm{ or }k > n.
\end{equation*}
Next, for any partitions $\lambda, \mu\in\Y(N)$, define
\begin{equation}\label{eqn:tau}
\tau_{\lambda/\mu}(t; N) := (-t^{1-N})^{|\lambda| - |\mu|} {N - \mu_1' \brack \lambda_1' - \mu_1'}\prod_{i \geq 1}{\mu_i' - \mu_{i+1}' \brack \lambda_{i+1}' - \mu_{i+1}'}.
\end{equation}
Note that $\tau_{\lambda/\mu}(t; N) = 0$ unless $\mu\subseteq\lambda$ and 	$\lambda/\mu$ is a vertical strip.

\begin{prop}\label{prop:expansionHL}
For any $\lambda\in\Y(N)$, we have
\begin{equation*}
\F_{\lambda}(x_1, \ldots, x_N; t) = \sum_{\mu}{\tau_{\lambda/\mu}(t; N)\PP_{\mu}(x_1, \ldots, x_N; t)}.
\end{equation*}
\end{prop}

We need some preparatory lemmas; in the rest of this subsection, $N$ is a fixed positive integer.
First, we slightly extend the definition of Hall-Littlewood polynomial.
An $N$-tuple $l = (l_1, \ldots, l_N)\in\N_0^N$ is called an \textit{almost-partition} if
\begin{itemize}
	\item for each $i = 1, 2, \ldots, N-1$, either $l_i \geq l_{i+1}$, or $l_i = l_{i+1} - 1$;
	\item there do not exist indices $1\leq i < j < k \leq N$ with $l_i < l_j < l_k$.
\end{itemize}
For such $l\in\N_0^N$, define
\begin{eqnarray}
m_k(l) &:=& \#\{ 1\leq i\leq N : l_i = k \} \ \forall k \geq 0,\nonumber\\
\inv(l) &:=& \#\{ 1\leq i < j \leq N : l_i > l_j \},\nonumber\\
v_l(t) &:=& \frac{\prod_{i \geq 0}{\phi_{m_i(l)}(t)}}{(1-t)^N} = \prod_{i\geq 0}\prod_{j=1}^{m_i(l)} \frac{1 - t^j}{1 - t},\label{vDef}\\
\PP_{l}(x_1, \ldots, x_N; t) &:=&  v_l(t)^{-1}\sum_{w\in S_N}
w\left\{ x_1^{l_1}\cdots x_N^{l_N}\prod_{1\leq i<j\leq N}\frac{x_i - tx_j}{x_i - x_j} \right\}.\label{Pdef}
\end{eqnarray}

Given an almost-partition $l\in\N_0^N$, let $\lambda\in\Y(N)$ be the partition given by
\begin{equation*}
\lambda := (1^{m_1(l)}2^{m_2(l)}\dots).
\end{equation*}
We say that $\lambda\in\Y(N)$ is the partition \textit{linked to} $l\in\N_0^N$.
In particular, we have
\begin{equation*}
m_k(l) = m_k(\lambda) \ \forall k \geq 1, \ m_0(l) = N - \ell(\lambda), \hspace{.2in} v_{\lambda}(t) = v_l(t),
\end{equation*}
and the $N$-tuple $\lambda = (\lambda_1, \ldots, \lambda_N)$ is obtained after applying $\inv(l)$ transpositions to $l = (l_1, \ldots, l_N)$.

\begin{lem}\label{lem:straight}
Let $l\in\N_0^N$ be an almost-partition and $\lambda\in\Y(N)$ be the partition linked to $l$, as defined above. Then
\begin{equation*}
\PP_{l}(x_1, \ldots, x_N; t) = t^{\inv(l)}\PP_{\lambda}(x_1, \ldots, x_N; t).
\end{equation*}
\end{lem}
\begin{proof}
This is a very special case of \cite[Lemma]{Mo}. See also \cite{L}.
In fact, the Lemma in \cite{Mo} gives a Hall-Littlewood polynomial expansion for $\PP_{l}(x_1, \ldots, x_N; t)$ and any $N$-tuple $l$ of nonnegative integers, being the definition $(\ref{Pdef})$ extended in the obvious way. Such expansion is more complicated in the general case, but it simplifies greatly for almost-partitions.
\end{proof}

Next let $\lambda\in\Y(N)$ be a partition with $\ell(\lambda) = k\leq N$, and $I\subseteq\{1, 2, \ldots, k\}$.
Consider the almost-partition $p\in\N_0^N$, defined from $(\lambda, I)$, via
\begin{equation*}
p := (\lambda_1 - \mathbf{1}_{\{1\in I\}}, \ldots, \lambda_k - \mathbf{1}_{\{k\in I\}}, \underbrace{0, \ldots, 0}_{N-k \textrm{ zeroes}}).
\end{equation*}
We write $p = \Pi(\lambda, I)$ for this dependence.
We denote the partition linked to $p = \Pi(\lambda, I)$ by $\pi(\lambda, I)$.
Some evident relations are
\begin{equation}\label{eqn:evident}
|I| = |\lambda| - |\pi(\lambda, I)|, \hspace{.5in} v_{\Pi(\lambda, I)}(t) = v_{\pi(\lambda, I)}(t).
\end{equation}

\begin{lem}\label{lem:verticalstrip}
Let $\lambda, \mu\in\Y(N)$ be arbitrary, and $\ell(\lambda) = k\leq N$.
There exists $I\subseteq\{1, \ldots, k\}$ such that $\pi(\lambda, I) = \mu$ if and only if $\mu\subseteq\lambda$ and $\lambda/\mu$ is a vertical strip.
\end{lem}
\begin{proof}
Let $d := \lambda_1$, and for each $1\leq j\leq d$, let $X_j := \{ i : \lambda_i = j\}$.
Then $|X_j| = m_j(\lambda)$. Also, for $I\subseteq\{1, \ldots, k\}$, let $I_j := I \cap X_j$ and $i_j := |I_j|$, for each $1\leq j\leq d$, so that $0 \leq i_j \leq m_j(\lambda)$ and $i_1 + \ldots + i_d = |I|$.
If $\mu = \pi(\lambda, I)$, then the construction of the map $\pi$ implies
\begin{align*}
m_d(\mu) &= m_d(\lambda) - i_d,\\
m_{d-1}(\mu) &= m_{d-1}(\lambda) - i_{d-1} + i_d,\\
&\cdots\\
m_1(\mu) &= m_1(\lambda) - i_1 + i_2.
\end{align*}
From $m_d(\lambda) = \lambda_d'$, $m_i(\lambda) = \lambda_i' - \lambda_{i+1}'$ for $i < d$, and the analogous relations for $\mu$, we deduce
\begin{align*}
i_d &= \lambda_d' - \mu_d',\\
i_{d-1} &= \lambda_{d-1}' - \mu_{d-1}',\\
&\cdots\\
i_1 &= \lambda_1' - \mu_1'.
\end{align*}
The bounds $0\leq i_j \leq m_j(\lambda) = \lambda_j' - \lambda_{j+1}'$ then yield the interlacing relation
\begin{equation}\label{interlacing}
\lambda_1' \geq \mu_1' \geq \lambda_2' \geq \dots \geq \lambda_{d-1}' \geq \mu_{d-1}' \geq \lambda_d'.
\end{equation}
In particular, this implies that if $\mu = \pi(\lambda, I)$ for some $I \subseteq\{1, \ldots, k\}$, then $\mu\subseteq\lambda$ and $\lambda/\mu$ is a vertical strip.
Conversely, if $\lambda/\mu$ is a vertical strip, then the interlacing relation $(\ref{interlacing})$ is satisfied. Then the previous argument shows that any $I\subseteq\{1, \ldots, k\}$ with $|I \cap X_j| = \lambda_j' - \mu_j'$ gives $\pi(\lambda, I) = \mu$.
\end{proof}

\begin{lem}\label{lem:combinatorial}
Let $\lambda\in\Y(N)$ be arbitrary with $\ell(\lambda) = k\leq N$, and let $\mu\in\Y(N)$ be such that $\lambda/\mu$ is a vertical strip.
Then
\begin{equation}\label{eqn:toproveproduct}
\sum_{\substack{I\subseteq\{1, \ldots, k\} : \\ \pi(\lambda, I) = \mu}}t^{\inv(\Pi(\lambda, I))} = \prod_{i \geq 1}
{\lambda_i' - \lambda_{i+1}' \brack \mu_i' - \lambda_{i+1}'}.
\end{equation}
\end{lem}
\begin{proof}
We begin with the simplest case, which is $\lambda = (a^k) = (\underbrace{a, \ldots, a}_{k\textrm{ times}})$, for some $a \geq 1$.
Then given any $I \subseteq \{1, 2, \ldots, k\}$, say $I = \{i_1 < \ldots < i_s\}$, we clearly have $a' := \Pi(\lambda, I) = (a - \mathbf{1}_{\{1\in I\}}, \ldots, a - \mathbf{1}_{\{k\in I\}}, 0^{N-k})$, and $\{(i, j) : i < j, \ a'_i > a_j'\} = \{(i, j) : i < j, \ i\in I, \ j \notin I\}$. Thus
\begin{equation*}
\inv(a') = \inv(\Pi(\lambda, I)) = \sum_{r=1}^s{(k - i_r + r - s)}.
\end{equation*}
Next, the only $\mu\in\Y(N)$ such that $\lambda/\mu$ is a vertical strip are those of the form $\mu = (a, \ldots, a, a-1, \ldots, a-1)$. Let us say that $\mu$ has $(k - s)$ entries that are $a$ and $s$ entries that are $a-1$. Then the sets $I$ such that $\pi(\lambda, I) = \mu$ are exactly those of size $s$; it follows that
\begin{equation*}
\sum_{I : \pi(\lambda, I) = \mu}t^{\inv(\Pi(\lambda, I))} = \sum_{1\leq i_1 < \ldots < i_s \leq k}{t^{\sum_{r=1}^s{(k - i_r + r - s)}}} = \sum_{0\leq j_1 \leq \ldots \leq j_s \leq k-s}{t^{\sum_{r=1}^s{(k - j_r - s)}}}.
\end{equation*}
The latter sum equals $h_s(1, t, \ldots, t^{k-s})$, where $h_r$ is the $r$-th complete homogeneous symmetric polynomial. There is an explicit formula for this evaluation, see e.g. \cite[Ch. I.3, Ex. 1]{M}, and it yields the desired result:
\begin{equation*}
\sum_{I : \pi(\lambda, I) = \mu}t^{\inv(\Pi(\lambda, I))} = {k \brack s}.
\end{equation*}

For a general partition $\lambda\in\Y(N)$ of length $k$, let us use the notation of the previous lemma.
Let $d := \lambda_1$, and for each $1\leq j\leq d$, let $X_j := \{ i : \lambda_i = j\}$.
Then $|X_j| = m_j(\lambda)$. Also, for $I\subseteq\{1, \ldots, k\}$, let $I_j := I \cap X_j$ and $i_j := |I_j|$, for each $1\leq j\leq d$, so that $0 \leq i_j \leq m_j(\lambda)$ and $i_1 + \ldots + i_d = |I|$.

Denote also $a = (a_1, \ldots, a_N) := \Pi(\lambda, I) = (\lambda_1 - \mathbf{1}_{\{1\in I\}}, \ldots, \lambda_k - \mathbf{1}_{\{k\in I\}}, 0^{N-k})$.
It is clear that if $(i, j)$ is a pair such that $1\leq i < j \leq N$, $a_i < a_j$, then $i, j\in X_r$ for some $r$.
Thus, if we let $s_j+1$ be the smallest element of $X_j$, so that $X_j = \{s_j + 1, \ldots, s_j + m_j\}$, and $t_j$ the number of involutions needed to transform $(a_{s_j + 1}, \ldots, a_{s_j + m_j})$ into a partition,
it follows that the number of involutions needed to transform $a$ into a partition is the sum $t_1 + \ldots + t_d$.
By definition, this is denoted as $\inv(\Pi(\lambda, I)) = t_1 + \ldots + t_d$.
On the other hand, it is clear that $t_j$ is also the number of involutions needed to transform $(1 - \mathbf{1}_{\{s_j + 1 \in I_j\}}, \ldots, 1 - \mathbf{1}_{\{s_j + m_j\in I_j\}})$ into a partition.
Therefore
\begin{equation}\label{breakinv}
\inv(\Pi(\lambda, I)) = \sum_{j=1}^d{\inv(\Pi((1^{m_j}), I_j - s_j))},
\end{equation}
where $(1^{m_j})$ is the column partition of length $m_j$ and $I_j - s_j := \{i - s_j : i\in I_j\} \subseteq \{1, 2, \ldots, m_j\}$.
	
To proceed, observe that by inspecting the proof of Lemma $\ref{lem:verticalstrip}$, $\pi(\lambda, I) = \mu$ if and only if
\begin{align*}
|I_d| &= \lambda_d' - \mu_d',\\
|I_{d-1}| &= \lambda_{d-1}' - \mu_{d-1}',\\
&\cdots\\
|I_1| &= \lambda_1' - \mu_1'.
\end{align*}
Thus the sum in the left side of $(\ref{eqn:toproveproduct})$ is over those $I \subseteq \{1, \ldots, k\}$ with $|I_j| = \lambda_j' - \mu_j'$.
Combining this observation with $(\ref{breakinv})$ and the case previously considered, we obtain
\begin{equation*}
\sum_{\substack{I\subseteq\{1, \ldots, k\} : \\ \pi(\lambda, I) = \mu}}t^{\inv(\Pi(\lambda, I))} =
\prod_{j=1}^d \sum_{\substack{I_j \subseteq \{1, \ldots, m_j(\lambda)\} \\ |I_j| = \lambda_j' - \mu_j'}}t^{\inv(\Pi((1^{m_j(\lambda)}), I_j - s_j))} = \prod_{j=1}^d
{m_j(\lambda) \brack \lambda_j' - \mu_j'}.
\end{equation*}
Since $m_j(\lambda) = \lambda_j' - \lambda_{j+1}'$, the lemma follows.
\end{proof}

\begin{proof}[Proof of Proposition $\ref{prop:expansionHL}$]
For $\lambda\in\Y(N)$, let $v_{\lambda}(t)$ be the factor in front of the (inhomogeneous) HL polynomials (same formula as in $(\ref{vDef})$).
Then, from the definition $(\ref{Pdef})$, we have
\begin{gather*}
\F_{\lambda}(x_1, \ldots, x_N; t) = v_{\lambda}(t)^{-1}\sum_{w\in S_N}w\left\{ \prod_{i=1}^k{(1 - t^{1-N}x_i^{-1})} \prod_{i=1}^k{x_i^{\lambda_i}}
\prod_{1\leq i<j\leq N}\frac{x_i - tx_j}{x_i - x_j} \right\}\\
= v_{\lambda}(t)^{-1}\sum_{w\in S_N}w\left\{ \sum_{I\subseteq\{1, \ldots, k\}}{(-t^{1-N})^{|I|}} \prod_{i=1}^k{x_i^{\lambda_i - \mathbf{1}_{\{i\in I\}}}}
\prod_{1\leq i<j\leq N}\frac{x_i - tx_j}{x_i - x_j} \right\}\\
= v_{\lambda}(t)^{-1} \sum_{I\subseteq\{1, \ldots, k\}} {(-t^{1-N})^{|I|}} \sum_{w\in S_N} w\left\{ \prod_{i=1}^k{x_i^{\lambda_i - \mathbf{1}_{\{i\in I\}}}}
\prod_{1\leq i<j\leq N}\frac{x_i - tx_j}{x_i - x_j} \right\}\\
= v_{\lambda}(t)^{-1} \sum_{I\subseteq\{1, \ldots, k\}} {(-t^{1-N})^{|I|}}\cdot v_{\Pi(\lambda, I)}(t) \PP_{\Pi(\lambda, I)}(x_1, \ldots, x_N; t).
\end{gather*}

From Lemma $\ref{lem:straight}$, we have
\begin{equation*}
\PP_{\Pi(\lambda, I)}(x_1, \ldots, x_N; t) = t^{\inv(\Pi(\lambda, I))}\PP_{\pi(\lambda, I)}(x_1, \ldots, x_N; t),
\end{equation*}
and therefore, by using also $(\ref{eqn:evident})$, we obtain
\begin{gather*}
\F_{\lambda}(x_1, \ldots, x_N; t) = v_{\lambda}(t)^{-1} \sum_{I\subseteq\{1, \ldots, k\}} (-t^{1-N})^{|I|}\cdot v_{\Pi(\lambda, I)}(t)  t^{\inv(\Pi(\lambda, I))}\PP_{\pi(\lambda, I)}(x_1, \ldots, x_N; t)\\
= v_{\lambda}(t)^{-1} \sum_{\mu} \PP_{\mu}(x_1, \ldots, x_N; t) \sum_{I : \pi(\lambda, I) = \mu}
(-t^{1-N})^{|I|}\cdot v_{\Pi(\lambda, I)}(t)  t^{\inv(\Pi(\lambda, I))}\\
= v_{\lambda}(t)^{-1} \sum_{\mu} \PP_{\mu}(x_1, \ldots, x_N; t) (-t^{1-N})^{|\lambda| - |\mu|} v_{\mu}(t) \sum_{I : \pi(\lambda, I) = \mu}
t^{\inv(\Pi(\lambda, I))}.
\end{gather*}
To finish the proof of the proposition, it then suffices to show
\begin{equation*}
(-t^{1-N})^{|\lambda| - |\mu|} v_{\lambda}(t)^{-1} v_{\mu}(t) \sum_{I : \pi(\lambda, I) = \mu} t^{\inv(\Pi(\lambda, I))} = \tau_{\lambda/\mu}(t; N).
\end{equation*}

From Lemma $\ref{lem:combinatorial}$, the definition of $\tau_{\lambda/\mu}(t; N)$, and the definition $(\ref{vDef})$, the latter equation is equivalent to
\begin{equation}\label{toprove1}
\frac{\phi_{N - \mu_1'}(t)}{\phi_{N - \lambda_1'}(t)}\prod_{i \geq 1}{\frac{\phi_{\mu_i' - \mu_{i+1}'}(t)}{\phi_{\lambda_i' - \lambda_{i+1}'}(t)}} \prod_{i \geq 1}\frac{\phi_{\lambda_i' - \lambda_{i+1}'}(t)}{\phi_{\mu_i' - \lambda_{i+1}'}(t)\phi_{\lambda_i' - \mu_i'}(t)} = {N - \mu_1' \brack \lambda_1' - \mu_1'}\prod_{i \geq 1}{\mu_i' - \mu_{i+1}' \brack \mu_i' - \lambda_{i+1}'},
\end{equation}
which we now check.
On the left side of $(\ref{toprove1})$, we can cancel the factors $\phi_{\lambda_i' - \lambda_{i+1}'}(t)$, $i \geq 1$, from the numerator and denominator.
Next, by using $\prod_{i \geq 1}\phi_{\lambda_i' - \mu_i'}(t) = \phi_{\lambda_1' - \mu_1'}(t)\prod_{i \geq 1}\phi_{\lambda_{i+1}' - \mu_{i+1}'}(t)$, the left side of $(\ref{toprove1})$ can be written as
\begin{equation*}
\frac{\phi_{N - \mu_1'}(t)}{\phi_{N - \lambda_1'}(t)\phi_{\lambda_1' - \mu_1'}(t)}\prod_{i \geq 1} \frac{\phi_{\mu_i' - \mu_{i+1}'}(t)}{\phi_{\mu_i' - \lambda_{i+1}'}(t)\phi_{\lambda_{i+1}' - \mu_{i+1}'}(t)} = {N - \mu_1' \brack \lambda_1' - \mu_1'}\prod_{i \geq 1}{\mu_i' - \mu_{i+1}' \brack \mu_i' - \lambda_{i+1}'}.
\end{equation*}
\end{proof}

\subsection{A degeneration of interpolation Macdonald polynomials}\label{sec:Hall2}

Let us begin with the Hall-Littlewood degeneration of the interpolation Macdonald polynomials, see \cite{Ol}.

\begin{prop}
\begin{enumerate}
	\item For any $\mu\in\Y(N)$, there exists a limit
\begin{equation*}
\Fcal_{\mu}(x_1, \ldots, x_N; t) = \lim_{q \rightarrow 0} I_{\mu | N}(x_1, \ldots, x_N; 1/q, 1/t),
\end{equation*}
which is a polynomial.
	\item One has the following combinatorial formula
\begin{equation*}
\Fcal_{\mu}(x_1, \ldots, x_N; t) = \sum_{T\in\Tab(\mu, N)} \psi_T(t)\prod_{(i, j)\in\mu}{(x_{T(i, j)} - \delta_{j, 1}t^{T(i, j) - N - i + 1})},
\end{equation*}
where $\Tab(\mu, N)$ is the set of semistandard Young tableaux of shape $\mu$, filled with numbers in the set $\{1, 2, \ldots, N\}$, and each $\psi_T(t)$ is the weight of the tableau $T$ that shows up in the combinatorial formula for Hall-Littlewood polynomials; see \cite[Ch. III, (5.9')]{M}.
\end{enumerate}
\end{prop}
\begin{proof}
This is proved in \cite[Lem. 9.2]{Ol}.
\end{proof}

Next we degenerate the well known binomial formula for interpolation Macdonald polynomials, \cite{Ok3}, in the limit regime $q\rightarrow 0$.
We need several lemmas.

\begin{lem}\label{limitspec}
For any $\mu\in\Y(N)$,
\begin{equation*}
\lim_{q\rightarrow 0} q^{2n(\mu') + |\mu|}\cdot I_{\mu|N}(q^{-\mu_1}, q^{-\mu_2}t, \ldots, q^{-\mu_N}t^{N-1}; q, t) = t^{n(\mu)}.
\end{equation*}
\end{lem}
\begin{proof}
From $(\ref{eqn:alternative})$ and $(\ref{Hexp})$, we have
\begin{equation*}
\begin{gathered}
q^{2n(\mu') + |\mu|}I_{\mu|N}(q^{-\mu_1}, q^{-\mu_2}t, \ldots, q^{-\mu_N}t^{N-1}; q, t)\\
= q^{2n(\mu') + |\mu|}C(\mu; q, t) = t^{n(\mu)}\prod_{s\in\mu}{(1 - q^{a(s)+1}t^{l(s)})}
\end{gathered}
\end{equation*}
and the result follows because $\lim_{q\rightarrow 0}{q^{a(s)+1}} = 0$.
\end{proof}

\begin{lem}\label{limitzeroes}
For any $\mu\in\Y(N)$, we have
\begin{equation*}
\lim_{q\rightarrow 0} q^{-n(\mu')} I_{\mu|N}(0^N; 1/q, 1/t) = (-t^{1-N})^{|\mu|}{N \brack \mu_1'}\prod_{i \geq 1}{\mu_i' \brack \mu_{i+1}'}.
\end{equation*}
\end{lem}
\begin{proof}
From $(\ref{evalzeros})$, with $(q, t)$ replaced by $(1/q, 1/t)$, we obtain
\begin{eqnarray*}
q^{-n(\mu')} I_{\mu|N}(0^N; 1/q, 1/t) &=& (-1)^{|\mu|}t^{-2n(\mu)}\prod_{s\in\mu} { \frac{q^{-a'(s)}t^{l'(s) - N} - 1}{q^{-a(s)}t^{-l(s)-1} - 1} }\\
&=& (-1)^{|\mu|}t^{-2n(\mu)}\prod_{s\in\mu} { t^{1-N + l'(s) + l(s)} \frac{1- q^{a'(s)}t^{N - l'(s)}}{1 - q^{a(s)}t^{l(s) + 1}} }\\
&=& (-t^{1-N})^{|\mu|}\prod_{s\in\mu} { \frac{1- q^{a'(s)}t^{N - l'(s)}}{1 - q^{a(s)}t^{l(s) + 1}} },
\end{eqnarray*}
where the middle equality holds because of the equalities $\sum_{s\in\mu}{a(s)} = \sum_{s\in\mu}{a'(s)}$, whereas the last one holds because $\sum_{s\in\mu}{l(s)} = \sum_{s\in\mu}{l'(s)} = n(\mu)$.
Therefore
\begin{equation}\label{eqn:limitzero}
\lim_{q\rightarrow 0} q^{-n(\mu')} I_{\mu|N}(0^N; 1/q, 1/t) = (-t^{1-N})^{|\mu|}\prod_{s\in\mu}\frac{1 - t^{N - l'(s)}\mathbf{1}_{\{a'(s) = 0\}}}{1 - t^{l(s)+1}\mathbf{1}_{\{a(s) = 0\}}}.
\end{equation}

The coarm length $a'(s)$ vanishes if and only if $s = (i, 1)$, for $i = 1, \ldots, \mu_1'$. Therefore $\prod_{s\in\mu}{(1 - t^{N - l'(s)}\mathbf{1}_{\{a'(s) = 0\}})} = \prod_{i=1}^{\mu_1'}{(1 - t^{N - i + 1})} = (1-t)^{\mu_1'}\cdot [N]!/[N - \mu_1']!$.

On the other hand, the arm length $a(s)$ vanishes if and only if $s = (i, \mu_i)$, for $i = 1, \ldots, \mu_1'$.
Let $\{1, 2, \ldots, \mu_1'\} = X_1\sqcup X_2\sqcup\dots$, where $X_k := \{1 \leq i\leq \mu_1' : \mu_i = k \}$; observe that $|X_k| = m_k(\mu) = \mu_k' - \mu_{k+1}'$.
Clearly $\prod_{s = (i, j) : i\in X_k}{(1 - t^{l(s)+1}\mathbf{1}_{\{a(s) = 0\}})} = \prod_{s = (i, \lambda_i) : i\in X_k}{(1 - t^{l(s)+1})} = (1-t)(1-t^2)\cdots (1-t^{|X_k|}) = (1-t)^{|X_k|}\cdot [|X_k|]! = (1-t)^{\mu_k' - \mu_{k+1}'}\cdot [\mu_k' - \mu_{k+1}']!$.
Then $\prod_{s\in\mu}{(1 - t^{l(s)+1}\mathbf{1}_{\{a(s)=0\}})} = (1-t)^{\mu_1'}\prod_{k\geq 1}{[\mu_k' - \mu_{k+1}']!}$.

Therefore, from $(\ref{eqn:limitzero})$ and the previous simplifications:
\begin{equation*}
\lim_{q\rightarrow 0} q^{-n(\mu')} I_{\mu|N}(0^N; 1/q, 1/t) = \frac{(-t^{1-N})^{|\mu|}\cdot [N]!}{[N-\mu_1']!\prod_{k\geq 1}{[\mu_k' - \mu_{k+1}']!}} = (-t^{1-N})^{|\mu|}{N \brack \mu_1'}\prod_{i \geq 1}{\mu_i' \brack \mu_{i+1}'}.
\end{equation*}
\end{proof}

For the next limiting statement, Lemma $\ref{limitcomb}$, we need a preparatory lemma.

\begin{lem}\label{lem:bound1}
Let $\mu\subseteq\lambda$ be two partitions of length $\leq N$. Let $T$ be a semistandard Young tableau of shape $\mu$ and filled with numbers in $\{1, 2, \ldots, N\}$. Then
\begin{equation}\label{ineq1}
n(\mu') + n(\lambda') + |\mu| - \sum_{(i, j)\in\mu}{\max(\lambda_{T(i, j)},\ j - 1)} \geq 0,
\end{equation}
and equality holds if and only if $\lambda/\mu$ is a vertical strip, and $\lambda_{T(i, j)} = \lambda_i$ for all $(i, j)\in\mu$.
\end{lem}
\begin{proof}
By using
\begin{gather*}
n(\mu') = \sum_{i=1}^N{ {\mu_i \choose 2} }, \hspace{.1in} n(\lambda') = \sum_{i=1}^N{ {\lambda_i \choose 2} },
\end{gather*}
the statement of the lemma is equivalent to
\begin{equation}\label{ineq2}
{\mu_i \choose 2} + {\lambda_i \choose 2} + \mu_i - \sum_{j=1}^{\mu_i}{\max(\lambda_{T(i, j)}, j - 1)} \geq 0 \ \forall i = 1, 2, \ldots, N,
\end{equation}
with equality if and only if $\lambda_i \in \{\mu_i, \mu_i + 1\}$ and $\lambda_{T(i, j)} = \lambda_i$.

Since $T(i, 1) \leq T(i, 2) \leq \cdots \leq T(i, \mu_i)$, we have $\lambda_{T(i, 1)} \geq \lambda_{T(i, 2)} \geq \cdots \geq \lambda_{T(i, \lambda_i)}$. Thus $\lambda_{T(i, j)}$ is decreasing in $j$, whereas $j - 1$ is increasing. This implies there exists $0 \leq a_i \leq \mu_i$ such that
\begin{equation*}
\max(\lambda_{T(i, j)}, j - 1) = \begin{cases}
\lambda_{T(i, j)},& \textrm{ if }j \leq a_i,\\
j -1, & \textrm{ if }j > a_i.
\end{cases}
\end{equation*}
Then we deduce
\begin{equation}\label{ineq3}
\sum_{j=1}^{\mu_i}{\max(\lambda_{T(i, j)}, j - 1)} = \sum_{j=1}^{a_i}{\lambda_{T(i, j)}} + \sum_{j=a_i+1}^{\mu_i}{(j-1)} \leq a_i\lambda_i + {\mu_i \choose 2} - {a_i \choose 2},
\end{equation}
because $T(i, j) > T(i-1, j) > \dots > T(1, j)$ implies $T(i, j) \geq i$ and $\lambda_{T(i, j)} \leq \lambda_i$; equality holds if and only if $\lambda_i = \lambda_{T(i, j)}$ for all $1\leq j\leq a_i$.

From $(\ref{ineq3})$, the left side of $(\ref{ineq2})$ multiplied by two is at least equal to
\begin{equation}\label{somebound}
\lambda_i(\lambda_i - 1) + 2\mu_i - 2a_i\lambda_i + a_i(a_i - 1) = (\lambda_i - a_i)^2 + (2\mu_i - \lambda_i - a_i).
\end{equation}
Since $\mu\subseteq\lambda$, then $\lambda_i \geq \mu_i \geq a_i$, which implies that $(\ref{somebound})$ is lower bounded by
\begin{equation*}
(\lambda_i - a_i)^2 + (2\mu_i - \lambda_i - a_i) \geq (\lambda_i - a_i)^2 -(\lambda_i - a_i) = (\lambda_i - a_i)(\lambda_i - a_i - 1) \geq 0,
\end{equation*}
and equality holds if and only if $\mu_i = a_i$ and $\lambda_i = \mu_i$ or $\lambda_i = \mu_i + 1$. Putting everything together, the lemma is proved.
\end{proof}

\begin{lem}\label{limitcomb}
For any $\mu, \lambda\in\Y(N)$ with $\mu\subseteq\lambda$, we have
\begin{equation}\label{toprovecomb}
\lim_{q\rightarrow 0}{q^{n(\mu') + n(\lambda') + |\mu|}}\cdot I_{\mu|N}(q^{-\lambda_1}, q^{-\lambda_2}t, \ldots, q^{-\lambda_N}t^{N-1}; q, t) = t^{n(\mu)} \prod_{i \geq 1}{{\lambda_i' - \lambda_{i+1}'} \brack \lambda_i' - \mu_i'}.
\end{equation}
In particular, this limit is zero unless $\lambda/\mu$ is a vertical strip.
\end{lem}

\begin{proof}
\textbf{Step 1.}
The combinatorial formula for interpolation Macdonald polynomials, see \cite[Thm. III]{Ok}, is 
\begin{equation*}
I_{\mu | N}(q^{-\lambda_1}, \ldots, q^{-\lambda_N}t^{N-1}; q, t) = \sum_{T \in\Tab(\mu, N)}{\psi_T(q, t) \prod_{(i, j)\in\mu}{(q^{-\lambda_{T(i, j)}}t^{T(i, j)-1} - q^{1-j}t^{N-1+i-T(i, j)})}}.
\end{equation*}
We want to show that for all tableaux $T\in\Tab(\mu, N)$, the limit
\begin{equation}\label{limittab}
\lim_{q\rightarrow 0}
q^{n(\mu') + n(\lambda') + |\mu|}\prod_{(i, j)\in\mu}{(q^{-\lambda_{T(i, j)}}t^{i-1} - q^{1-j}t^{N-1+i-T(i, j)})}
\end{equation}
exists and we want to determine conditions on $\lambda \supseteq \mu$ and $T\in\Tab(\mu, N)$ for which the limit is nonzero.
Write the prelimit expression above as
\begin{equation}\label{prelimit2}
\begin{gathered}
q^{n(\mu') + n(\lambda') + |\mu| - \sum_{(i, j)\in\mu}{\max(\lambda_{T(i, j)}, j-1)}}\\
\times\prod_{(i, j)\in\mu}{(q^{\max(0,\ j-1-\lambda_{T(i, j)})}t^{i-1} - q^{\max(0, \ \lambda_{T(i, j)} - j + 1)}t^{N-1+i-T(i, j)})}.
\end{gathered}
\end{equation}
It is clear that the second line of $(\ref{prelimit2})$ has a limit as $q\rightarrow 0$; in fact, this limit is
\begin{equation}
\prod_{(i, j)\in\mu}{\left( \mathbf{1}_{\{ \lambda_{T(i, j)} + 1 \geq j \}} t^{i-1} - \mathbf{1}_{\{ j \geq \lambda_{T(i, j)} - 1 \}} t^{N-1+i-T(i, j)}\right)}.
\end{equation}
From Lemma $\ref{lem:bound1}$, the first line of $(\ref{prelimit2})$ has a limit as $q\rightarrow 0$, and that limit is zero unless $\lambda/\mu$ is a vertical strip and $\lambda_{T(i, j)} = \lambda_i$ for all $(i, j)\in\mu$.

So far, we have proved that if $\lambda/\mu$ is not a vertical strip, then
\begin{equation*}
\lim_{q\rightarrow 0}{q^{n(\mu') + n(\lambda') + |\mu|}}\cdot I_{\mu|N}(q^{-\lambda_1}, q^{-\lambda_2}t, \ldots, q^{-\lambda_N}t^{N-1}; q, t) = 0.
\end{equation*}
Moreover, if $\lambda/\mu$ is a vertical strip, then
\begin{equation}\label{limitexp2}
\begin{gathered}
\lim_{q\rightarrow 0}{q^{n(\mu') + n(\lambda') + |\mu|}}\cdot I_{\mu|N}(q^{-\lambda_1}, q^{-\lambda_2}t, \ldots, q^{-\lambda_N}t^{N-1}; q, t)\\
= \sum_{\substack{T \in\Tab(\mu, N) \\ \lambda_{T(i, j)} = \lambda_i \ \forall (i, j)\in\mu}}{\psi_T(t) \prod_{(i, j)\in\mu}{\left( \mathbf{1}_{\{ \lambda_{T(i, j)} + 1 \geq j \}} t^{T(i, j)-1} - \mathbf{1}_{\{ j \geq \lambda_{T(i, j)} - 1 \}} t^{N-1+i-T(i, j)}\right)} }.
\end{gathered}
\end{equation}
It remains to simplify the last expression. Assume in the remainder of the proof that $\lambda/\mu$ is a vertical strip.\\

\textbf{Step 2.} If $T\in\Tab(\mu, N)$ satisfies $\lambda_{T(i, j)} = \lambda_i$ for any $(i, j)\in\mu$, then $\lambda_{T(i, j)} + 1 = \lambda_i + 1 > j$.
Thus each factor simplifies as $( \mathbf{1}_{\{ \lambda_{T(i, j)} + 1 \geq j \}} t^{T(i, j)-1} - \mathbf{1}_{\{ j \geq \lambda_{T(i, j)} - 1 \}} t^{N-1+i-T(i, j)}) = t^{T(i, j)-1}$.
Therefore, the second line of $(\ref{limitexp2})$ is simplified to
\begin{equation}\label{sumphis}
\sum_{\substack{T \in\Tab(\mu, N) \\ \lambda_{T(i, j)} = \lambda_i \ \forall (i, j)\in\mu}}{\psi_T(t)\cdot \prod_{(i, j)\in\mu}{t^{T(i, j) - 1}}}.
\end{equation}
In the remaining steps we show that $(\ref{sumphis})$ equals the right side of $(\ref{toprovecomb})$.
More explicitly, we show in Step 3 that for any $T\in\Tab(\mu, N)$ with $\lambda_{T(i, j)} = \lambda_i$ for all $(i, j)\in\mu$, one has $\psi_T(t) = 1$. Finally, in Step 4, it is shown that $\sum_{T}{\prod_{(i, j)\in\mu}t^{T(i, j)-1}}$, the sum being over $T\in\Tab(\mu, N)$ with $\lambda_{T(i, j)} = \lambda_i$, equals the right side of $(\ref{toprovecomb})$.\\

\textbf{Step 3.}
Let $T \in \Tab(\mu, N)$ be such that $\lambda_{T(i, j)} = \lambda_i$ for all $(i, j)\in\mu$.
We recall the definition of $\psi_T(t)$.
The tableau $T$ is given by a sequence $\mu = \mu^{(N)} \succ \mu^{(N-1)} \succ \dots \succ \mu^{(1)} \succ \mu^{(0)} = \emptyset$, where $\mu^{(k)}$ is the set of boxes of $\mu$ filled with numbers $\leq k$.
Given $\nu \succ \kappa$, $\theta := \nu - \kappa$ is a horizontal strip; set
\begin{equation}\label{psihorizontal}
\psi_{\nu/\kappa}(t) := \prod_{j\in J}{(1 - t^{m_j(\kappa)})},
\end{equation}
where $J := \{j : \theta_j' = 0, \ \theta_{j+1}' = 1\}$.
Then by definition
\begin{equation*}
\psi_T(t) := \prod_{i=1}^N{\psi_{\mu^{(i)}/\mu^{(i-1)}}(t)}.
\end{equation*}
For $T$ as described above, and any $i \geq 1$, we will argue that $\psi_{\mu^{(i)}/\mu^{(i-1)}}(t) = 1$; this will show $\psi_T(t)=1$.
For any $1\leq k\leq \mu_1'$, let
\begin{equation*}
Y_k := \{ i : \mu_i = k, \ \lambda_i = k+1\}, \ Z_k := \{i : \mu_i = k = \lambda_i\}.
\end{equation*}
Since $\lambda/\mu$ is a vertical strip, we deduce
\begin{equation*}
|Y_k| = \lambda_{k+1}' - \mu_{k+1}', \ |Z_k| = \mu_k' - \lambda_{k+1}'.
\end{equation*}

We claim that for any $i\in Y_k$, then $T(i, j) = i$.
In fact, $\{i : \lambda_i = k+1\} = Y_k \sqcup Z_{k+1}$.
Say $Z_{k+1} = \{z_1 < \ldots < z_r\}$, $Y_k = \{y_1 < \ldots < y_r\}$, so that $z_r < y_1$.
By definition of Young tableau, $T(z_1, j) < \ldots < T(z_r, j) < T(y_1, j) < \ldots < T(y_r, j)$.
But by assumption $T(z_1, j), \ldots, T(z_r, j), T(y_1, j), \ldots, T(y_r, j) \in Y_k \sqcup Z_{k+1} = \{z_1 < \ldots < z_r < y_1 < \ldots < y_r\}$.
In particular, we have $T(y_i, j) = y_i$ for all $i$, i.e., $T(i, j) = i$ for any $i\in Y_k$, as claimed.

By a similar reasoning as above, we deduce: if $i\in Z_k$ and $j \neq k$, then $T(i, j) = i$; and if $i\in Z_k$ and $Y_{k-1} = \emptyset$, then also $T(i, j) = i$.

From the claims above we see that, possibly, the only boxes $(i, j)\in\mu$ with $T(i, j) \neq i$ are those with $j = \mu_i = \lambda_i$ and for which there exist $k > i$ with $\lambda_k = \mu_k + 1 = j$.
For a fixed $j$, there exist $\mu_j' - \lambda_{j+1}'$ such boxes $(i, j)$, and the numbers on those boxes can be chosen from the set $\{\lambda_{j+1}' + 1, \lambda_{j+1}' + 2, \ldots, \lambda_j'\}$ in such a way that they are strictly increasing from bottom to top.
See Figure $\ref{fig:young}$ for an illustration in the case $\mu = (6, 5, 5, 4, 4, 4, 2, 2, 1)$, $\lambda = (7, 5, 5, 5, 5, 5, 2, 2, 1, 1, 1)$.

From these considerations, it is clear that $\psi_{\mu^{(i)}/\mu^{(i-1)}}(t) = 1$ for any $i$, because each set $J$ in the definition $(\ref{psihorizontal})$ is the empty set. This is what we wished to prove.\\

\begin{figure}
\begin{center}
\begin{ytableau}
1& 1& 1& 1& 1& 1& *(battleshipgrey)\\
2& 2 & 2 &2 & *(ashgrey)i\\
3& 3& 3& 3& *(ashgrey)j\\
4& 4& 4& 4& *(battleshipgrey)\\
5& 5& 5& 5& *(battleshipgrey)\\
6& 6& 6& 6\\
7& 7\\
8& 8\\
*(ashgrey)k\\
*(battleshipgrey)\\
*(battleshipgrey)\\
\end{ytableau}
\caption{$\mu = (6, 4, 4, 4, 4, 4, 2, 2, 1)$, $\lambda = (7, 5, 5, 5, 5, 4, 2, 2, 1, 1, 1)$.
Filled squares (with numbers or letters) belong to $\mu$, whereas the dark-grey squares belong to $\lambda/\mu$.
Any $T \in \Tab(\mu, 15)$ with $\lambda_{T(i, j)} = \lambda_i$ has $T(i, j) = i$ for most squares $(i, j)\in\mu$. The numbers have been written in those squares.
For the light-grey squares $(i, j)\in\mu$, $T(i, j)$ could take one of several values.
In our example, $i < j$ take values in $\{2, 3, 4, 5\}$, whereas $k$ takes values in $\{9, 10, 11\}$.}
\label{fig:young}
\end{center}
\end{figure}

\textbf{Step 4.}
From the previous step, $T(i, j) = i$, unless $j = \mu_i = \lambda_i$ and $\lambda_{j-1}' - \lambda_j' > 0$.
In the latter case, $T(i, j) - i$ could be any number in the set $\{0, 1, \ldots, \lambda_{j-1}' - \lambda_j' - 1\}$;
we note also that for a given $j$, there are $\mu_j' - \lambda_{j+1}'$ boxes like these.

Since $\prod_{(i, j)\in\mu}{t^{i-1}} = t^{\mu_2 + 2\mu_3 + \dots} = t^{n(\mu)}$, we deduce
\begin{equation*}
\sum_{\substack{T \in\Tab(\mu, N) \\ \lambda_{T(i, j)} = \lambda_i \ \forall (i, j)\in\mu}}{\prod_{(i, j)\in\mu}{t^{T(i, j) - 1}}}
= t^{n(\mu)}\prod_{j\geq 1}{\sum_{0\leq k_1 \leq \ldots \leq k_{\mu_j' - \lambda_{j+1}'}\leq \lambda_{j-1}' - \lambda_j' - 1}{t^{k_1 + \ldots + k_{\mu_j' - \lambda_{j+1}'}}}}.
\end{equation*}
The inner sum above can be calculated:
\begin{equation*}
\sum_{0\leq k_1 \leq \ldots \leq k_{\mu_j' - \lambda_{j+1}'}\leq \lambda_{j-1}' - \lambda_j' - 1}{t^{k_1 + \ldots + k_{\mu_j' - \lambda_{j+1}'}}} = h_{\mu_j' - \lambda_{j+1}'}(1, t, \ldots, t^{\lambda_j' - \lambda_{j+1}' - 1}) = {\lambda_j' - \lambda_{j+1}' \brack \mu_j' - \lambda_{j+1}'}
\end{equation*}
as in the proof of Lemma $\ref{lem:combinatorial}$; the proof is now finished.
\end{proof}

\begin{prop}\label{prop:Fcal}
For any $\lambda\in\Y(N)$, we have
\begin{equation*}
\Fcal_{\lambda}(x_1, \ldots, x_N; t) = \sum_{\mu}{\tau_{\lambda/\mu}(t; N)\PP_{\mu}(x_1, \ldots, x_N; t)},
\end{equation*}
where the expressions $\tau_{\lambda/\mu}(t; N)$ were defined in $(\ref{eqn:tau})$.
\end{prop}

\begin{proof}
From the binomial formula for interpolation Macdonald polynomials in \cite{Ok3} with $(1/q, 1/t)$ instead of $(q, t)$, and using the well-known symmetry $M_{\mu}(x; q, t) = M_{\mu}(x; 1/q, 1/t)$, we obtain
\begin{equation}\label{binomialbefore}
\begin{gathered}
I_{\lambda | N}(x_1, \ldots, x_N; 1/q, 1/t)
= \sum_{\mu \subseteq \lambda}
\left\{\frac{q^{-n(\lambda')}I_{\lambda | N}(0^N; 1/q, 1/t)}{q^{-n(\mu')}I_{\mu | N}(0^N; 1/q, 1/t)}\right.\\
\left.\frac{q^{n(\mu') + n(\lambda') + |\mu|} I_{\mu | N}(q^{-\lambda_1}, \ldots, q^{-\lambda_N}t^{N-1}; q, t)}{q^{2n(\mu')+|\mu|} I_{\mu | N}(q^{-\mu_1}, \ldots, q^{-\mu_N}t^{N-1}; q, t)}
M_{\mu | N}(x_1, \ldots, x_N; q, t)\right\}.
\end{gathered}
\end{equation}

We know
\begin{gather*}
\lim_{q\rightarrow 0} I_{\lambda | N}(x_1, \ldots, x_N; 1/q, 1/t) = \Fcal_{\lambda}(x_1, \ldots, x_N; t),\\
\lim_{q\rightarrow 0} M_{\mu | N}(x_1, \ldots, x_N; q, t) = \PP_{\mu}(x_1, \ldots, x_N; t).
\end{gather*}
From Lemmas $\ref{limitspec}$, $\ref{limitzeroes}$ and $\ref{limitcomb}$, the binomial formula in $(\ref{binomialbefore})$ has a limit as $q$ tends to zero.
The limiting coefficient that accompanies $\PP_{\mu}(x_1, \ldots, x_N; t)$ is
\begin{gather*}
\frac{(-t^{1-N})^{|\lambda|}{N \brack \lambda_1'}\prod_{i \geq 1}{\lambda_i' \brack \lambda_{i+1}'}}{(-t^{1-N})^{|\mu|}{N \brack \mu_1'}\prod_{i \geq 1}{\mu_i' \brack \mu_{i+1}'}}
\times\frac{t^{n(\mu)} \prod_{i \geq 1}{{\lambda_i' - \lambda_{i+1}'} \brack \lambda_i' - \mu_i'}}{t^{n(\mu)}},
\end{gather*}
which is easily seen to be equal to $\tau_{\lambda/\mu}(t; N)$.
\end{proof}

From Propositions $\ref{prop:expansionHL}$ and $\ref{prop:Fcal}$, we obtain the main result of this section:

\begin{thm}\label{cor:equalityHLs}
For any $\lambda\in\Y(N)$, we have
\begin{equation*}
\F_{\lambda}(x_1, \ldots, x_N; t) = \Fcal_{\lambda}(x_1, \ldots, x_N; t) = \sum_{T\in\Tab(\lambda, N)} \psi_T(t)\prod_{(i, j)\in\lambda}{(x_{T(i, j)} - \delta_{j, 1}t^{T(i, j) - N - i + 1})}.
\end{equation*}
\end{thm}

\subsection{Special cases: one column partition and one row partition}

For any $k\geq 1$, let
\begin{equation*}
\EE_k(x_1, \ldots, x_N; t) := 
	\begin{cases}
            1, & \text{ if }k = 0,\\
            \F_{(1^k)}(x_1, \ldots, x_N; t), & \text{ if } 1 \leq k \leq N.
        \end{cases}
\end{equation*}
The polynomials $\EE_k(x_1, \ldots, x_N; t)$ are inhomogeneous analogues of the elementary symmetric polynomials $e_k(x_1 \ldots, x_N) = \sum_{1\leq i_1 < \ldots < i_k \leq N}{x_{i_1}\cdots x_{i_k}}$.

\begin{prop}
For any $0\leq k\leq N$, we have
\begin{equation}\label{eqn:elementary}
\EE_k(x_1, \ldots, x_N; t)= I_{(1^k)|N}(x_1, \ldots, x_N; q^{-1}, t^{-1}) = \sum_{1\leq i_1 < \dots < i_k \leq N}{\prod_{s=1}^k{(x_{i_s} - t^{s+1-k-i_s})}}.
\end{equation}
Consequently, if we denote the symmetric function $\EE_k(\cdot; t) := \F_{(1^k)}(\cdot; t)$, then $\EE_k(\cdot; t) = I_{(1^k) |N}(\cdot; q^{-1}, t^{-1})$.
\end{prop}
\begin{proof}
The second equality in $(\ref{eqn:elementary})$ is in the article of Okounkov \cite[(1.6)]{Ok3}. Recall that Okounkov's notation and ours, for interpolation polynomials, are related by $I_{\mu|N}(x_1, \ldots, x_N; q, t) = P_{\mu}^*(x_1, x_2/t, \ldots, x_N/t^{N-1}; 1/q, 1/t)$.

Note that $I_{(1^k) | N}(x_1, \ldots, x_N; q^{-1}, t^{-1})$ does not depend on $q$.
Thus Theorem $\ref{cor:equalityHLs}$ yields the first equality $\EE_k(x_1, \ldots, x_N; t) = \F_{(1^k)}(x_1, \ldots, x_N; t) = I_{(1^k) | N}(x_1, \ldots, x_N; q^{-1}, t^{-1})$.
\end{proof}

\begin{cor}
\begin{equation*}
\sum_{k = 0}^N{\frac{\EE_k(x_1, \ldots, x_N; t)}{(u+1)(u+t^{-1})\cdots (u + t^{1-k})}}
= \prod_{i=1}^N{\frac{1 + x_i/u}{1 + t^{1-i}/u}}.
\end{equation*}
\end{cor}
\begin{proof}
This is a rewriting of \cite[(2.9)]{Ok3}. Let us observe that we can replace $N$ by $\infty$ in the upper limit of the sum and product, if we also replace $\EE_k(x_1, \ldots, x_N; t)$ by $\EE_k(\cdot; t)$.
\end{proof}

Next, let us denote
\begin{equation*}
\HH_k(x_1, \ldots, x_N; t) := \F_{(k)}(x_1, \ldots, x_N; t), \text{ if } k \geq 0,
\end{equation*}
in particular $\HH_0(x_1, \ldots, x_N; t) = 1$. The polynomials $\HH_k(x_1, \ldots, x_N; t)$ are inhomogeneous analogues of complete homogeneous symmetric polynomials $h_k(x_1 \ldots, x_N) = \sum_{1\leq i_1 \leq \ldots \leq i_k \leq N}{x_{i_1}\cdots x_{i_k}}$.

From Proposition $\ref{prop:expansionHL}$, we have
\begin{equation}\label{eqn:complete}
\begin{gathered}
\HH_1(x_1, \ldots, x_N; t) = \PP_{(1)}(x_1, \ldots, x_N; t) - \frac{t^{1-N}(1 - t^N)}{1 - t},\\
\HH_k(x_1, \ldots, x_N; t) = \PP_{(k)}(x_1, \ldots, x_N; t) - t^{1-N}\PP_{(k-1)}(x_1, \ldots, x_N; t), \textrm{ for } k\geq 2.
\end{gathered}
\end{equation}
From \cite[Ch. III, (2.10)]{M}, we have the generating series
\begin{equation}\label{eqn:generatingHall}
1 + (1-t)\sum_{n=1}^{\infty}{\PP_{(n)}(x_1, \ldots, x_N; t)u^n} = \prod_{i=1}^N{\frac{1 - x_itu}{1 - x_iu}}.
\end{equation}

\begin{prop}
The following is the generating series for $\HH_n(x_1, \ldots, x_N; t)$:
\begin{equation}\label{eqn:generatingH}
\sum_{n=0}^{\infty}{\HH_n(x_1, \ldots, x_N; t)u^n} = \frac{1 - t^{1-N}u}{1 - t}\prod_{i=1}^N{\frac{1 - x_i tu}{1 - x_iu}} - \frac{t(1 - u)}{1 - t},
\end{equation}
and consequently, if we denote $\HH_n(\cdot; t) := \F_{(n)}(\cdot. t)$, we have
\begin{equation}\label{eqn:generatingfun}
\sum_{n=0}^{\infty}{\HH_n(\cdot; t)u^n} = \frac{1}{1 - t}\prod_{i=1}^{\infty}{\frac{1 - x_i tu}{1 - x_iu}} - \frac{t(1 - u)}{1 - t}.
\end{equation}
\end{prop}
\begin{proof}
The generating function $(\ref{eqn:generatingH})$ is a consequence of $(\ref{eqn:generatingHall})$ and $(\ref{eqn:complete})$.

To obtain $(\ref{eqn:generatingfun})$ from $(\ref{eqn:generatingH})$, informally, treat $t$ as a real number with $t > 1$ and send $N$ to infinity. This agrees with the remarks made before Proposition/Definition $\ref{df:Halls}$, regarding the construction of the maps $\pi_n^{\infty}$ as limits of the maps $\pi^{N}_{N-1}\circ \dots \circ \pi^{n+1}_n$.
Formally, one can write the right side of $(\ref{eqn:generatingfun})$ in terms of the power sums $\{p_n : n\geq 1\}$ and the right side of $(\ref{eqn:generatingH})$ in terms of the set of generators $\{p_n(x_1, \ldots, x_N) : 1\leq n\leq N\}$.
One then checks that after replacing each $p_n$ by $\pi^{\infty}_N p_n = p_n(x_1, \ldots, x_N) + t^{-nN}/(1 - t^{-n})$, the right side of $(\ref{eqn:generatingfun})$ becomes the right side of $(\ref{eqn:generatingH})$, cf. the argument in $(\ref{gather:eqns})$.

One can also derive $(\ref{eqn:generatingfun})$ from \cite[(2.10)]{Ok3} and Theorem $\ref{cor:equalityHLs}$.
\end{proof}

\subsection{Vertex operator representation for the first operator}

As an application of the material from the previous subsection, we give a formula for the operator $A^1$, other than that in equation $(\ref{eqn:kexplicit})$.

From $(\ref{eqn:generatingfun})$, we obtain
\begin{equation}\label{eqn:generating1}
\begin{gathered}
\sum_{n=1}^{\infty}{\HH_n(\cdot; t^{-1})u^n} = \frac{t}{t - 1}
\exp\left(\sum_{i=1}^{\infty}\{ -\ln{(1 - x_iu)} + \ln{(1 - x_iu/t)} \}\right) + \frac{t - u}{1 - t}\\
= \frac{t}{t - 1}
\exp\left(\sum_{n=1}^{\infty}  \frac{u^n(1 - t^{-n})}{n}p_n  \right) + \frac{t - u}{1 - t}.
\end{gathered}
\end{equation}

Next, to the equation $(\ref{eqn:generatingHall})$, set $t^{-1}x_i$ instead of $x_i$, then $t^{-1}$ instead of $t$, and finally send $N$ to infinity. By recalling $\QQ_{(n)}(\cdot; t^{-1}) = (1 - t^{-1})\PP_{(n)}(\cdot; t^{-1})$, we have
\begin{gather*}
\sum_{n=1}^{\infty}{t^n \QQ_{(n)}(\cdot; t^{-1}) u^n} = \exp\left( \sum_{i=1}^{\infty}\{ -\ln{(1 - x_itu)} + \ln{(1 - x_iu)} \} \right) - 1\\
= \exp\left( \sum_{n=1}^{\infty}{ \frac{u^n(t^n - 1)}{n}p_n } \right) - 1.
\end{gather*}
From the definition $(\ref{eqn:macdonaldinner})$ of the Macdonald inner product, we have $p_n^* = n\frac{1 - q^n}{1 - t^n}\frac{\partial}{\partial p_n}$. After taking the adjoint of the last equation and using $u^{-1}$ instead of $u$, we then obtain
\begin{equation}\label{eqn:generating2}
\sum_{n=1}^{\infty}{t^n (\QQ_{(n)}(\cdot; t^{-1}))^* u^{-n}} = \exp\left( \sum_{n=1}^{\infty}{ u^{-n}(q^n - 1)\frac{\partial}{\partial p_n} } \right) - 1.
\end{equation}

From Theorem $\ref{thm:main}$ for $k = 1$, the operator $A^1$ is the constant coefficient of the product of generating functions on the left sides of $(\ref{eqn:generating1})$ and $(\ref{eqn:generating2})$. We deduce the following vertex operator representation for $A^1$. It would be interesting to obtain similar formulas for all operators $A^k$.

\begin{prop}
\begin{equation*}
A^1 = \frac{t}{t - 1}\oint_{|z| \ll 1}
\frac{dz}{2\pi\sqrt{-1} z}
\exp\left(\sum_{n=1}^{\infty}  \frac{z^n(1 - t^{-n})}{n}p_n  \right)
\exp\left( \sum_{n=1}^{\infty}{ z^{-n}(q^n - 1)\frac{\partial}{\partial p_n} } \right)
- \frac{t}{t - 1} + \frac{1 - q}{1 - t}\frac{\partial}{\partial p_1}
\end{equation*}
\end{prop}

\subsection{Another relation to Hall-Littlewood polynomials}

\begin{prop}\label{prop:inhHLs}
Let $N\in\N$; for any $1\leq j\leq N$, let $T_{x_j}$ be the operator $(T_{x_j}f)(x_1, \ldots, x_N) := f(x_1, \ldots, x_{i-1}, 0, x_{i+1}, \ldots, x_N)$. Then
\begin{equation}\label{eqn:inhHLs}
\begin{gathered}
\F_{\lambda}(x_1, \ldots, x_N; t) = \frac{(u; t^{-1})_{\ell(\lambda)}}{(u; t^{-1})_N\cdot V(x_1, \ldots, x_N)}\\
\times\det_{1\leq i, j\leq N}\left[ x_j^{N - i - 1}\left\{ (1 - x_j u) t^{1-i}T_{x_j} + (x_j - t^{1-N}) \right\} \right] \PP_{\lambda}(x_1, \ldots, x_N; t)
\end{gathered}
\end{equation}
\end{prop}

Observe that the left side of $(\ref{eqn:inhHLs})$  does not depend on the variable $u$, therefore neither does the right side of that equality. In particular, by setting $u = 0$ we have
\begin{equation*}
\begin{gathered}
\F_{\lambda}(x_1, \ldots, x_N; t) = \frac{1}{V(x_1, \ldots, x_N)}\det_{1\leq i, j\leq N}\left[ x_j^{N - i - 1}\left\{ t^{1-i}T_{x_j} + x_j - t^{1-N} \right\} \right] \PP_{\lambda}(x_1, \ldots, x_N; t),
\end{gathered}
\end{equation*}
which together with $(\ref{eqn:Hallstability})$ gives an explicit formula for $\F_{\lambda}(x_1, \ldots, x_N; t)$ in terms of the set of HL polynomials $\{\PP_{\lambda}(\{x_i : i\in I\}; t)\}_{I \subseteq \{1, \ldots, N\}, |I| \geq \ell(\lambda)}$.

\begin{proof}[Proof of Proposition $\ref{prop:inhHLs}$]
Begin with Proposition $\ref{prop:refined}$, after setting $x_i \mapsto x_i/(ut^{N-1})$, $y_i \mapsto y_iut^{N-1}$, and using the homogeneity of the dual HL polynomials:
\begin{equation}\label{eqn:tocompare}
\begin{gathered}
1 + \sum_{k=1}^N{\frac{1}{(u; t^{-1})_k}\sum_{\ell(\lambda) = k}{(ut^{N-1})^{\lambda_1 + \lambda_2 + \dots}\F_{\lambda}(\frac{x_1}{ut^{N-1}}, \dots, \frac{x_N}{ut^{N-1}}; t)\QQ_{\lambda}(y_1, y_2, \dots; t)}} =\\
\frac{u^{-\frac{N(N-1)}{2}}t^{-\frac{N(N-1)^2}{2}}}{V(x_1, \dots, x_N)(u; t^{-1})_N}\det_{1\leq i, j\leq N}
\left[ x_j^{N - i - 1}\left\{ (x_j t^{1-N} - 1) t^{N-i} (-u) + (x_j - u)\prod_{l = 1}^{\infty}{\frac{1 - tx_jy_l}{1 - x_jy_l}} \right\} \right].
\end{gathered}
\end{equation}

Next consider the operator in $(\ref{eqn:Dop})$ with $u = -z$, namely
\begin{equation}\label{eqn:modifiedop}
D_N(u; q, t) := \frac{1}{V(x_1, \ldots, x_N)}\det_{1\leq i, j\leq N}\left[ x_j^{N - i - 1}\left\{ (x_j t^{1 - N} - 1)t^{N - i} (-u) T_{q, x_j} + (x_j - u) \right\} \right].
\end{equation}
We apply this operator to the \textit{Cauchy identity for Macdonald polynomials}, \cite[Ch. VI, (4.13)]{M}
\begin{equation}\label{eqn:cauchyMacs}
\sum_{\ell(\lambda) \leq N}
M_{\lambda}(x_1, \ldots, x_N; q, t)\cdot MQ_{\lambda}(y_1, y_2, \ldots; q, t) =
\prod_{j=1}^N\prod_{l = 1}^{\infty}{\frac{(tx_jy_l; q)_{\infty}}{(x_jy_l; q)_{\infty}}},
\end{equation}
where $MQ_{\lambda}(\cdot; q, t)$ stands for the dual Macdonald function, which differs from $M_{\lambda}(\cdot; q, t)$ by a constant not depending on $\lambda$ and, more importantly for us, it specializes to $\QQ_{\lambda}(\cdot; t)$ when we set $q = 0$. With a calculation that is similar to that of $(\ref{eqn:leftside})$, we can write down an expression for the conjugation of the operator $(\ref{eqn:modifiedop})$ by the right side of $(\ref{eqn:cauchyMacs})$.
Then we deduce that the result of acting with $(\ref{eqn:modifiedop})$ on $(\ref{eqn:cauchyMacs})$, and then dividing by $(u; t^{-1})_N$, is
\begin{equation}\label{eqn:beforeqzero}
\begin{gathered}
\sum_{\ell(\lambda) \leq N}
\frac{D_N(u; q, t) M_{\lambda|N}(x_1, \ldots, x_N; q, t)}{(u; t^{-1})_N} \cdot MQ_{\lambda}(y_1, y_2, \ldots; q, t) =
\frac{\prod_{j=1}^N\prod_{l = 1}^{\infty}{\frac{(tx_jy_l; q)_{\infty}}{(x_jy_l; q)_{\infty}}}}{V(x_1, \ldots, x_N)\cdot (u; t^{-1})_N}\\
\times \det_{1\leq i, j\leq N}\left[ x_j^{N-i-1} \left\{ (x_j t^{1-N} - 1)t^{N-i} (-u) \prod_{l = 1}^{\infty}{\frac{1 - x_j y_l}{1 - tx_j y_l}} + (x_j - u) \right\} \right].
\end{gathered}
\end{equation}
We now want to set $q = 0$ in equation $(\ref{eqn:beforeqzero})$. As for the right side, only the product in the first line of that display is affected, and it becomes $\prod_{j=1}^N\prod_{l=1}^{\infty}{\frac{1 - tx_jy_l}{1 - x_jy_l}}$.
Then we can make the factor corresponding to $j$ multiply each term in the $j$-th column of the matrix in the second line of $(\ref{eqn:beforeqzero})$. Thus the result of setting $q = 0$ in $(\ref{eqn:beforeqzero})$ is
\begin{equation}\label{eqn:afterqzero}
\begin{gathered}
\sum_{\ell(\lambda) \leq N}
\frac{\left.D_N(u; q, t) M_{\lambda|N}(x_1, \ldots, x_N; q, t)\right|_{q=0}}{(u; t^{-1})_N} \cdot \QQ_{\lambda}(y_1, y_2, \ldots; t) =
\frac{1}{V(x_1, \ldots, x_N)\cdot (u; t^{-1})_N}\\
\times \det_{1\leq i, j\leq N}\left[ x_j^{N-i-1} \left\{ (x_j t^{1-N} - 1)t^{N-i} (-u) + (x_j - u)\prod_{l=1}^{\infty}{\frac{1 - tx_jy_l}{1 - x_jy_l}} \right\} \right].
\end{gathered}
\end{equation}

By comparing $(\ref{eqn:tocompare})$ and $(\ref{eqn:afterqzero})$, we deduce
\begin{equation}\label{eqn:FMpolys}
(ut^{N-1})^{\lambda_1 + \lambda_2 + \dots}\F_{\lambda}(\frac{x_1}{ut^{N-1}}, \dots, \frac{x_N}{ut^{N-1}}; t) = 
\frac{(u; t^{-1})_{\ell(\lambda)}}{(u; t^{-1})_N}\left.D_N(u; q, t) M_{\lambda|N}(x_1, \ldots, x_N; q, t)\right|_{q=0}.
\end{equation}

To obtain $(\ref{eqn:inhHLs})$, let us replace each $x_i$ by $x_iut^{N-1}$ in $(\ref{eqn:FMpolys})$.
From the homogeneity of Macdonald polynomials, we have $M_{\lambda}(x_1ut^{N-1}, \ldots, x_Nut^{N-1}; q, t) = (ut^{N-1})^{\lambda_1 + \lambda_2 + \dots}M_{\lambda}(x_1, \ldots, x_N; q, t)$, so that the factors $(ut^{N-1})^{\lambda_1 + \lambda_2 + \dots}$ on both sides of $(\ref{eqn:FMpolys})$ cancel out.
After the change, the operator $D_N(u; q, t)$ becomes almost the operator in the right side of $(\ref{eqn:inhHLs})$, except with $T_{q, x_j}$ instead of $T_{x_j}$, but one still has to set $q = 0$. Since clearly $\left.(T_{q, x_j}f)(x_1, \ldots, x_N)\right|_{q = 0} = (T_{x_j}f)(x_1, \ldots, x_N)$ for any function $f(x_1, \ldots, x_N)$, we obtain the desired result.

We lastly remark that the formula in the Proposition generalizes the special cases of Remark $\ref{rem:specialcases}$; they also serve as a check to our formula in the cases $\lambda = \emptyset$ and $\ell(\lambda) = N$.
\end{proof}

\bibliographystyle{amsplain}
\bibliography{PD}

\end{document}